\documentclass[12pt,draftcls,onecolumn]{IEEEtran}

\usepackage{graphicx}
\usepackage{array}
\usepackage{mathrsfs}
\usepackage{amsfonts}
\usepackage{setspace}
\usepackage{amsmath}
\usepackage{amssymb}
\usepackage{mathrsfs}
\input xy
\xyoption{all}

\newtheorem{theorem}{Theorem}
\newtheorem{lemma}{Lemma}

\newtheorem{definition}{Definition}

\newtheorem{remark}{Remark}
\newtheorem{example}{Example}
\newtheorem{problem}{Problem}
\newtheorem{algorithm}{Algorithm}

\begin{document}

 \title{Technical Report: NUS-ACT-11-003-Ver.1:\\
 Communicate only when necessary:\\ Cooperative tasking for multi-agent systems}
 \author{Mohammad~Karimadini,
        and~Hai~Lin
\thanks{M. Karimadini and H. Lin are both from the Department of Electrical and Computer Engineering, National University of Singapore,
Singapore. Corresponding author, H. Lin {\tt\small elelh@nus.edu.sg}
} } \maketitle \thispagestyle{empty} \pagestyle{empty}

\begin{abstract}
New advances in large scale distributed systems have amazingly
offered complex functionalities through parallelism of simple and
rudimentary components. The key issue in cooperative control of
multi-agent systems is the synthesis of local control and
interaction rules among the agents such that the entire controlled
system achieves a desired global behavior. For this purpose, three
fundamental problems have to be addressed: (1) task decomposition
for top-down design, such that the fulfillment of local tasks
guarantees the satisfaction of the global task, by the team; (2)
fault-tolerant top-down design, such that the global task remain
decomposable and achievable, in spite of some failures, and (3)
design of interactions among agents to make an undecomposable task
decomposable and achievable in a top-down framework.
 The first
two problems have been addressed in our previous works, by
identifying necessary and sufficient conditions for task automaton
decomposition, and fault-tolerant task decomposability, based on
decision making on the orders and selections of transitions,
interleaving of synchronized strings and determinism of bisimulation
quotient of local task automata. This paper deals with the third
problem and proposes a procedure to redistribute the events among
agents in order to enforce decomposability of an undecomposable task
automaton. The decomposability conditions are used to identify the
root causes of undecomposability which are found to be due to
over-communications that have to be deleted, while respecting the
fault-tolerant decomposability conditions; or because of the lack of
communications that require new sharing of events, while considering
new violations of decomposability conditions. This result provides a
sufficient condition to make any undecomposable deterministic task
automaton decomposable in order to facilitate cooperative tasking.
Illustrative examples are presented to show the concept of task
automaton decomposabilization.
\end{abstract}

%
%

\section{INTRODUCTION}
With new advances in technology and emergence of large scale complex
systems \cite{Lesser1999, Lima2005}, there is an ever-increasing
demand for cooperative control of distributed systems with
sophisticated specifications \cite{Choi2009, Kazeroon2009},
\cite{Georgios2009}, \cite{Ji2009} which impose new challenges that
fall beyond the traditional methods \cite{Tabuada2006, Belta2007,
Kloetzer2010, Georgios2009}. Conventional approaches either consider
the team of agents as a monolithic plant to be controlled by a
centralized unit, or design and iteratively adjust local
controllers, in a bottom-up structure, to generate a behavior closed
to a desired global behavior. Although the latter approache offers
more flexibility, scalability and functionality with lower cost, due
to local actuation and communications of agents \cite{Hinchey2005,
Truszkowski2006, Kloetzer2007}, they fail to guarantee a given
global specification \cite{Crespi2008}. For this purpose, top-down
cooperative control aims at formal design of local controllers in
order to collectively achieve the global specification, by design
\cite{Willner1991,cai2010supervisor}.

To address the top-down cooperative control, three fundamental
questions are evoked: The first question is the task decomposition
problem that is interested in understanding of whether all tasks are
decomposable, and if not, what are the conditions for task
decomposability. It furthermore asks that if the task is
decomposable and local controllers are designed to satisfy local
tasks, whether the whole closed loop system satisfies the global
specification. Subsequently, the second question refers to the
cooperative control under event failures, and would like to know if
after the task decomposition and local controller designs for global
satisfaction, some events fail in some agents, then whether the task
still remains decomposable and globally satisfied, in spite of event
failures. As another follow-up direction, the third question
investigates the way to make an undecomposable task decomposable
through modification of local agents in order to accomplish the
proposed cooperative control.

For cooperative control of logical behaviors \cite{Koutsoukos2000},
represented in automata \cite{Cassandras2008, Kumar1999}, the first
question (task decomposability for cooperative tasking) was
addressed in our previous work
\cite{Automatica2010-2-agents-decomposability}, by decomposing a
given global task automaton into two local task automata such that
their parallel composition bisimulates the original task automaton.
By using the notion of shared events, instead of common events and
incorporating the concept of global decision making on the orders
and selections between the transitions, the decomposability result
was generalized in \cite{TAC2011-n-agents-decomposability} to an
arbitrary finite number of agents. Given a deterministic task
automaton, and a set of local event sets, necessary and sufficient
conditions were identified for task automaton decomposability based
on decision making on the orders and selections of transitions,
interleaving of synchronized strings and determinism of bisimulation
quotient of local automata. It was also proven that the fulfillment
of local task automata guarantees the satisfaction of the global
specification, by design.

The second question, cooperative tasking under event failure, was
investigated in \cite{Automatica2011-Fault-tolerant}, by introducing
a notion of passive events to transform the fault-tolerant task
decomposability problem to the standard automaton decomposability
problem in \cite{TAC2011-n-agents-decomposability}. The passivity
was found to reflect the redundancy of communication links, based on
which the necessary and sufficient conditions have been then
introduced under which a previously decomposable task automaton
remains decomposable and achievable, in spite of events failures.
The conditions ensure that after passive failures, the team of
agents maintains its capability for global decision making on the
orders and selections between transitions; no illegal behavior is
allowed by the team (no new string emerges in the interleavings of
local strings) and no legal behavior is disabled by the team (any
string in the global task automaton appears in the parallel
composition of local automata). These conditions interestingly
guarantee the team of agents to still satisfy its global
specification, even if some local agents fail to maintain their
local specifications.

This paper deals with the third question to investigate how to make
undecomposable task automata decomposable in order for cooperative
tasking of multi-agent systems. For a global task automaton that is
not decomposable with respect to given local event sets, the problem
is particularly interested in finding a way to modify the local task
automata such that their parallel composition bisimulates the
original global task automaton, to guarantee its satisfaction by
fulfilling the local task automata.

Decomposition of different formalisms of logical specification have
been reported in the literature. Examples of such methods can be
seen for decomposition of a specification given in CSP
\cite{Moore1990}, decomposition of a LOTOS \cite{Go1999, Arora1998,
Hultstrom1994} and decomposition of petri nets \cite{Cao1994,
Zaitsev2004}. The problem of automaton decomposabilization has been
also studies in computer science literature. For example,
\cite{Morin1998} characterized the conditions for decomposition of
asynchronous automata in the sense of isomorphism based on the
maximal cliques of the dependency graph. The isomorphism equivalence
used in \cite{Morin1998} is however a strong condition, in the sense
that two isomorphic automata are bisimilar but not vise versa
\cite{Cassandras2008}. Moreover, \cite{Morin1998} considers a set of
events to be attributed to a number of agents, with no predefinition
of local event sets. While event attribution is suitable for
parallel computing and synthesis problems in computer science,
control applications typically deal with parallel distributed plants
\cite{Mukund2002} whose events are predefined by the set of sensors,
actuators and communication links across the agents. Therefore, it
would be advantageous to find a way to make an undecomposable
automaton decomposable with respect to predefined local event sets,
by modifying local task automata. Since the global task automaton is
fixed, one way to modify the local task automata is through the
modification in local event sets, which is the main theme of this
paper. Another related work is \cite{Kiyamura2003} that proposes a
method for automaton decomposabilization by adding synchronization
events such that the parallel composition of local automata is
observably bisimilar to the original automaton. The approach in
\cite{Kiyamura2003}, however, allows to add synchronization events
to the event set that will enlarge the size of global event set. Our
work deals with those applications with fixed global event sets and
predefined distribution of events among local agents, where
enforcing the decomposability is not allowed by adding the new
synchronization events, but instead by redistribution of the
existing events among the agents.

For this purpose, we propose an algorithm that uses previous results
on task decomposition \cite{Automatica2010-2-agents-decomposability,
TAC2011-n-agents-decomposability} to identify and overcome
dissatisfaction of each decomposability condition. The algorithm
first removes all redundant communication links using the
fault-tolerant result \cite{Automatica2011-Fault-tolerant}. As a
result, any violation of decomposability conditions, remained after
this stage, is not due to redundant communication links, and hence
cannot be removed by means of link deletions. Instead, the algorithm
proceeds by establishing new communication links to provide enough
information to facilitate the task automaton decomposition. Since
each new communication link may overcome several violations of
decomposability conditions, the algorithm may offer different
options for link addition, leading to the question of optimal
decomposability with minimum number of communication links.
It is found that if link additions impose no new violations of
decomposability conditions, then it is possible to make the
automaton decomposable with minimum number of links. However, it is
furthermore shown that, in general, addition of new communication
links may introduce new violations of decomposability conditions
that in turn require establishing new communication links. In such
cases, the optimal path depends on the structure of the automaton
and requires a dynamic exhaustive search to find the sequence of
link additions with minimum number of links. Therefore, in case of
new violations, a simple sufficient condition is proposed to provide
a feasible suboptimal solution to enforce the decomposability,
without checking of decomposability conditions after each link
addition. This approach can decompose any deterministic task
automaton, after which, according to the previous results, designing
local controllers such that local specification are satisfied,
guarantees the fulfillment of the global specification, by design.

The rest of the paper is organized as follows. Preliminary lemmas,
notations, definitions and problem formulation are represented in
Section \ref{PROBLEM FORMULATION}. This section also establishes the
links to previous works on task automaton decomposition and
fault-tolerant decomposition results. Section \ref{TASK AUTOMATON
DECOMPOSABILIZATION} proposes an algorithm to make any
undecomposable deterministic automaton decomposable by modifying its
local event sets.
Illustrative examples are also given to elaborate the concept of
task automaton decomposabilization.
Finally, the paper concludes with remarks and discussions in Section
\ref{CONCLUSION}. Proofs of the lemmas are readily given in the
Appendix.
\section{PROBLEM FORMULATION}\label{PROBLEM FORMULATION}
\subsection{Definitions and notations}
We first recall the definitions and notations used in this paper.

A \emph{deterministic automaton} is a tuple $A := \left(Q, q_0, E,
\delta \right)$ consisting of a set of states $Q$; an initial state
$q_0\in Q$; a set of events $E$ that causes transitions between the
states, and a transition relation $\delta \subseteq Q \times E\times
Q$, with partial map $\delta: Q \times E \to Q$, such that $(q, e,
q^{\prime})\in \delta$ if and only if state $q$ is transited to
state $q^{\prime}$ by event $e$, denoted by
$q\overset{e}{\underset{}\rightarrow } q^{\prime}$ (or $\delta(q, e)
= q^{\prime}$). A \emph{nondeterministic automaton} is a tuple $A :=
\left(Q, q_0, E, \delta \right)$ with a partial transition map
$\delta: Q \times E \to 2^Q$, and if hidden transitions
($\varepsilon$-moves) are also possible, then a nondeterministic
automaton with hidden moves is defined as $A := \left(Q, q_0,
E\cup\{\varepsilon\}, \delta \right)$ with a partial map $\delta: Q
\times (E\cup\{\varepsilon\} )\to 2^Q$. For a nondeterministic
automaton the initial state can be generally from a set
$Q_0\subseteq Q$. Given a nondeterministic automaton $A$, with
hidden moves, the $\varepsilon$-closure of $q\in Q$, denoted by
$\varepsilon^*_A(q)\subseteq Q$, is recursively defined as: $q\in
\varepsilon^*_A(q)$; $q^{\prime}\in \varepsilon^*_A(q)\Rightarrow
\delta(q^{\prime}, \varepsilon)\subseteq \varepsilon^*_A(q)$. The
transition relation can be extended to a finite string of events,
$s\in E^*$, where $E^*$ stands for $Kleene-Closure$ of $E$ (the set
of all finite strings over elements of $E$). For an automaton
without hidden moves, $\varepsilon^*_A(q) = \{q\}$, and the
transition on string is inductively defined as
$\delta(q,\varepsilon) = q$ (empty move or silent transition), and
$\delta(q,se)=\delta(\delta(q,s),e)$ for $s\in E^*$ and $e\in E$.
For an automaton $A$, with hidden moves, the extension of transition
relation on string, denoted by $\delta: Q \times E^*\to 2^Q$, is
inductively defined as: $\forall q \in Q, s \in E^*, e\in E$:
$\delta(q, \varepsilon):= \varepsilon^*_A(q)$ and $\delta(q,
se)=\varepsilon^*_A(\delta(\delta(q,s),e)) =  \overset{} {
\underset{q^{\prime}\in \delta(q, s)}{\cup} }\left\{
 \overset{}  { \underset{ q^{\prime\prime}\in \delta(q^{\prime}, e) } {\cup}
 }  \varepsilon^*_A(q^{\prime\prime}) \right\}$ \cite{Kumar1999}.

The operator $Ac(.)$ \cite{Cassandras2008} is then defined by
excluding the states and their attached transitions that are not
reachable from the initial state as $Ac(A) = \left(Q_{ac}, q_0, E,
\delta_{ac} \right)$ with $Q_{ac}=\{q\in Q|\exists s\in E^*, q \in
\delta (q_0, s)\}$ and $\delta_{ac}=\delta|Q_{ac}\times E\rightarrow
Q_{ac}$, restricting $\delta$ to the smaller domain of $Q_{ac}$.
Since $Ac(.)$ has  no effect on the behavior of the automaton, from
now on we take $A = Ac(A)$.

We focus on deterministic global task automata that are simpler to
be characterized, and cover a wide class of specifications. The
qualitative behavior of a deterministic system is described by the
set of all possible sequences of events starting from the initial
state. Each such a sequence is called a string, and the collection
of strings represents the language generated by the automaton,
denoted by $L(A)$. The existence of a transition over a string $s\in
E^*$ from a state $q\in Q$ is denoted by $\delta(q, s)!$.
Considering a language $L$, by $\delta(q, L)!$ we mean that $\forall
\omega \in L: \delta(q, \omega)!$. For $e\in E$, $s\in E^*$, $e\in
s$ means that $\exists t_1, t_2 \in E^*$ such that $s = t_1et_2$. In
this sense, the intersection of two strings $s_1, s_2 \in E^*$ is
defined as $s_1\cap s_2 = \{e|e \in s_1 \wedge e\in s_2\}$.
Likewise, $s_1\backslash s_2$ is defined as $s_1\backslash s_2 =
\{e|e\in s_1, e\notin s_2\}$. For $s_1, s_2 \in E^*$, $s_1$ is
called a sub-string of $s_2$, denoted by $s_1\leqslant s_2$, when
$\exists t\in E^*$, $s_2 = s_1t$. Two events $e_1$ and $e_2$ are
called successive events if $\exists q\in Q: \delta(q,e_1)! \wedge
\delta(\delta(q,e_1),e_2)!$ or $\delta(q,e_2)! \wedge
\delta(\delta(q,e_2),e_1)!$. Two events $e_1$ and $e_2$ are called
adjacent events if $\exists q\in Q:\delta(q,e_1)! \wedge
\delta(q,e_2)!$.

 To compare the task automaton and its decomposed automata, we use
the \emph{bisimulation relations}. Consider two automata $A_i=( Q_i,
q_i^0$, $E, \delta _i)$, $i=1, 2$. A relation $R\subseteq Q_1 \times
Q_2$ is said to be a simulation relation from $A_1$ to $A_2$ if
$(q_1^0, q_2^0) \in R$, and $\forall\left( {q_1 ,q_2 } \right) \in
R, \delta_1(q_1, e)= q'_1$, then $\exists q_2^{\prime}\in Q_2$ such
that $\delta_2(q_2, e)=q'_2, \left( {q'_1 ,q'_2 } \right) \in R$. If
$R$ is defined for all states and all events in $A_1$, then $A_1$ is
said to be similar to $A_2$ (or $A_2$ simulates $A_1$), denoted by
$A_1\prec A_2$ \cite{Cassandras2008}. If $A_1\prec A_2$, $A_2\prec
A_1$, with a symmetric relation, then $A_1$ and $A_2$ are said to be
bisimilar (bisimulate each other), denoted by $A_1\cong A_2$
\cite{Zhou2006}. In general, bisimilarity implies languages
equivalence but the converse does not necessarily hold
\cite{Alur2000}.

In these works \emph{natural projection} is used to obtain local
tasks, as local perspective of agents from the global task. Consider
a global event set $E$ and its local event sets $E_i$,
$i=1,2,...,n$, with $E=\overset{n}{\underset{i=1}{\cup} } E_i$.
Then, the natural projection $p_i:E^*\rightarrow E_i^*$ is
inductively defined as
 $p_i(\varepsilon)=\varepsilon$, and
 $\forall s\in E^*, e\in E:p_i(se)=\left\{
  \begin{array}{ll}
  p_i(s)e & \hbox{if $e\in E_i$;} \\
  p_i(s) & \hbox{otherwise.}
  \end{array}
\right.$ Accordingly, inverse natural projection $p_i^{-1}: E_i^*
\to 2^{E^*}$ is defined on an string $t\in E_i^*$ as $p_i^{-1}(t):=
\{s\in E^*|p_i(s) = t\}$.

The natural projection is also defined on automata as $P_i:
A\rightarrow A$, where, $A$ is the set of finite automata and
$P_i(A_S)$ are obtained from $A_S$ by replacing its events that
belong to $E\backslash E_i$ by $\varepsilon$-moves, and then,
merging the $\varepsilon$-related states. The  $\varepsilon$-related
states form equivalent classes defined as follows. Consider an
automaton $A_S=(Q, q_0, E, \delta)$ and a local event set
$E_i\subseteq E$. Then, the relation $\sim_{E_i}$ is the equivalence
relation on the set $Q$ of states such that $\delta(q, e) =
q^{\prime}\wedge e\notin E_i\Rightarrow q\sim_{E_i} q^{\prime}$, and
$[q]_{E_i}$ denotes the equivalence class of $q$ defined on
$\sim_{E_i}$.
 The set of equivalent classes of states
over $\sim_{E_i}$, is denoted by $Q_{/\sim_{E_i}}$ and defined as
$Q_{/\sim_{E_i}} = \{[q]_{E_i}|q\in Q\}$ \cite{Morin1998}. The
natural projection of $A_S$ into $E_i$ is then formally defined as
$P_i(A_S)=(Q_i = Q_{/\sim_{E_i}}, [q_0]_{E_i}, E_i, \delta_i)$, with
$\delta_i([q]_{E_i}, e) = [q^{\prime}]_{E_i}$ if there exist states
$q_1$ and $q_1^{\prime}$ such that $q_1\sim_{E_i} q$,
$q_1^{\prime}\sim_{E_i} q^{\prime}$, and $\delta(q_1, e) =
q^{\prime}_1$.

To investigate the interactions of transitions between automata,
particularly between $P_i(A_S)$, $i = 1, \ldots, n$, the
\emph{synchronized product of languages} is defined as follows.
Consider a global event set $E$ and local event sets $E_i$, $ i = 1,
\ldots, n$, such that $E = \overset{n}{\underset{i=1}{\cup} } E_i$.
For a finite set of languages $\{L_i \subseteq E_i^*\}_{i = 1}^n$,
the synchronized product (language product) of $\{L_i\}$, denoted by
$\overset{n}{\underset{i=1}{|} } L_i$, is defined as
$\overset{n}{\underset{i=1}{|} } L_i = \{s \in E^*|\forall
i\in\{1,\ldots, n\}: p_i(s)\in L_i\} =
\overset{n}{\underset{i=1}{\cap} } p_i^{-1}(L_i)$
\cite{Willner1991}.

Then, \emph{parallel composition (synchronized product)} is used to
define the composition of local task automata to retrieve the global
task automaton, and to model each local closed loop system by
compositions of its local plant and local controller automata. Let
$A_i=\left( Q_i,q_i^0,E_i,\delta _i\right)$, $i=1,2$ be automata.
The parallel composition (synchronous composition) of $A_1$ and
$A_2$ is the automaton $A_1||A_2=\left( Q = Q_1 \times Q_2, q_0 =
(q_1^0, q_2^0), E = E_1 \cup E_2, \delta\right)$, with $\delta$
defined as $\forall (q_1, q_2)\in Q$, $e\in E$: $\delta(\left(q_1,
q_2), e\right)=\\ \left\{
\begin{array}{ll}
    \left(\delta_1(q_1, e), \delta_2(q_2, e)\right), & \hbox{if $\left\{\begin{array}{ll}
    \delta _1(q_1,e)!, \delta _2(q_2,e)! \\
     e\in E_1 \cap E_2
\end{array}\right.$};\\
    \left(\delta_1(q_1, e), q_2\right), & \hbox{if $\delta _1(q_1,e)!, e\in E_1 \backslash E_2$;} \\
    \left(q_1, \delta_2(q_2, e)\right), & \hbox{if $\delta _2(q_2,e)!, e\in E_2 \backslash E_1$;} \\
    \hbox{undefined}, & \hbox{otherwise.}
\end{array}\right.$

The parallel composition of $A_i$, $i=1,2,...,n$ is called
\emph{parallel distributed system} (or concurrent system), and is
defined based on the associativity property of parallel composition
\cite{Cassandras2008} as $\overset{n}{\underset{i=1}{\parallel}
}A_i=A_1\parallel\ ...\parallel\   A_n = A_n\parallel \left(A_{n-1}
\parallel \left( \cdots \parallel \left( A_2\parallel
A_1 \right)\right)\right)$.

The set of labels of local event sets containing an event $e$ is
called \emph{the set of locations} of $e$, denoted by $loc(e)$ and
is defined as $loc(e) =\{i\in\{1,\ldots, n\}| e\in E_i\}$.

Based on these definitions, a task automaton $A_S$ with event set
$E$ and local event sets
 $E_i$, $i=1,..., n$, $E = \overset{n}{\underset{i=1}{\cup} } E_i$, is
said to be decomposable with respect to parallel composition and
natural projections $P_i$, $i=1,\cdots, n$, when
$\overset{n}{\underset{i=1}{\parallel} } P_i \left( A_S \right)
\cong A_S$.

\subsection{Problem formulation}

In \cite{Automatica2010-2-agents-decomposability}, we have shown
that not all automata are decomposable with respect to parallel
composition and natural projections, and subsequently necessary and
sufficient conditions were proposed for decomposability of a task
automaton with respect to parallel composition and natural
projections into two local event sets. These necessary and
sufficient conditions were then generalized to an arbitrary finite
number of agents, in \cite{TAC2011-n-agents-decomposability}, as
\begin{lemma} (Corollary $1$ in
\cite{TAC2011-n-agents-decomposability}) \label{Decomposability
Corollary for n agents} A deterministic automaton $A_S = \left( {Q,
q_0 , E = \bigcup\limits_{i = 1}^n {E_i , \delta } } \right)$ is
decomposable with respect to parallel composition and natural
projections $P_i$, $i=1,...,n$ such that $A_S \cong \mathop
{||}\limits_{i = 1}^n P_i \left( {A_S } \right)$ if and only if
$A_S$ satisfies the following decomposability conditions ($DC$):
\begin{itemize}
\item $DC1$: $\forall e_1,
e_2 \in E, q\in Q$: $[\delta(q,e_1)!\wedge \delta(q,e_2)!]\\
\Rightarrow [\exists E_i\in\{E_1, \ldots, E_n\}, \{e_1,
e_2\}\subseteq E_i]\vee[\delta(q, e_1e_2)! \wedge \delta(q,
e_2e_1)!]$;
\item $DC2$: $\forall e_1, e_2 \in E,  q\in Q$, $s\in E^*$: $[\delta(q,
e_1e_2s)!\vee \delta(q, e_2e_1s)!]\\ \Rightarrow [\exists
E_i\in\{E_1, \ldots, E_n\}, \{e_1, e_2\}\subseteq E_i]\vee [
\delta(q, e_1e_2s)!\wedge \delta(q, e_2e_1s)!]$;
\item $DC3$: $\delta(q_0, \mathop |\limits_{i = 1}^n p_i \left( {s_i }
\right))!$, $\forall \{s_1, \cdots, s_n\}\in \tilde L\left( {A_S }
\right)$, $\exists s_i, s_j \in \{s_1, \cdots, s_n\}$, $s_i \neq
s_j$, where, $\tilde L\left( {A_S } \right) \subseteq L\left( {A_S }
\right)$ is the largest subset of $L\left( {A_S } \right)$ such that
$\forall s\in \tilde L\left( {A_S } \right)\exists s^{\prime} \in
\tilde L\left( {A_S } \right),\;\exists E_i, E_j  \in \left\{ {E_1
,...,E_n } \right\},
 i \ne j,p_{E_i  \cap E_j } \left( s \right)$ and $p_{E_i  \cap E_j } \left( s^{\prime} \right)$  start with the same
 event, and
\item $DC4$: $\forall i\in\{1,...,n\}$, $x, x_1, x_2 \in Q_i$, $x_1\neq x_2$,
$e\in E_i$, $t\in E_i^*$, $\delta_i (x, e)=  x_1$, $\delta_i (x, e)=
x_2$: $\delta_i (x_1, t)! \Leftrightarrow \delta_i(x_2, t)!$.
 \end{itemize}
\end{lemma}

The first two decomposability conditions require the team to be
capable of decision on choice/order of events, by which for any such
decision there exists at least one agent that knows both events, or
the decision is not important. Moreover, the third and fourth
conditions, guarantee that the cooperative perspective of agents
from the tasks (parallel composition of local task automata) neither
allows a string that is prohibited by the global task automaton, nor
disables a string that is allowed in the global task automaton.

It was furthermore shown that once the task automaton is decomposed
into local task automata and local controllers are designed for
local plants to satisfy the local specifications, then the global
specification is guaranteed, by design.


 The next question
was the reliability of task decomposability to understand whether a
previously decomposable and achievable global task automaton, can
still remain decomposable and achievable by the team, after
experiencing some event failures. For this purpose, in
\cite{Automatica2011-Fault-tolerant}, \emph{a class of failures} was
investigated as follows to defined a notion of passivity. Consider
an automaton $A = (Q, q_0, E, \delta)$. An event $e\in E$ is said to
be failed in $A$ (or $E$), if $F(A) = P_{\Sigma}(A) = P_{E\backslash
e}(A) = (Q, q_0, \Sigma = E\backslash e, \delta^F)$, where,
$\Sigma$, $\delta^F$ and $F(A)$ denote the post-failure event set,
post-failure transition relation and post-failure automaton,
respectively. A set $\bar{E}\subseteq E$ of events is then said to
be failed in $A$, when for $\forall e\in \bar{E}$, $e$ is failed in
$A$, i.e., $F(A) = P_{\Sigma}(A_i) = P_{E\backslash \bar{E} }(A) =
(Q, q_0, \Sigma = E\backslash \bar{E}, \delta^F)$. Considering a
parallel distributed plant $A:= \overset{n}{\underset{i=1}{||} } A_i
= (Z, z_0, E = \overset{n}{\underset{i=1}{\cup} } E_i, \delta_{||})$
with local agents $A_i = (Q_i, q_0^i, E_i, \delta_i)$, $i = 1,
\ldots, n$. Failure of $e$ in $E_i$ is said to be \emph{passive} in
$E_i$ (or $A_i$) with respect to $\overset{n}{\underset{i=1}{||} }
A_i$, if $E = \overset{n}{\underset{i=1}{\cup} } \Sigma_i$. An event
whose failure in $A_i$ is a passive failure is called a passive
event in $A_i$.

The passivity was found to reflect the redundancy of communication
links and shown to be a necessary condition for preserving the
automaton decomposability. It was furthermore shown that when all
failed events are passive in the corresponding local event sets, the
problem of decomposability under event failure can be transformed
into the standard decomposability problem to find the conditions
under which $A_S \cong \mathop {||}\limits_{i = 1}^n
P_{E_i\backslash \bar{E}_i} (A_S )$, as follows.
\begin{lemma} (Theorem $1$ in \cite{Automatica2011-Fault-tolerant})\label{Decomposability under
event failure-Theorem} Consider a deterministic task automaton $A_S
= (Q, q_0, E = \mathop {\cup}\limits_{i = 1}^n E_i, \delta)$. Assume
that $A_S$ is decomposable, i.e., $A_S \cong \mathop {||}\limits_{i
= 1}^n P_i (A_S )$, and furthermore, assume that $\bar{E}_i =
\{a_{i,r}\}$ fail in $E_i$, $r\in\{1,...,n_i\}$, and $\bar{E}_i$ are
passive for $i\in\{1, \ldots, n\}$. Then, $A_S$ remains
decomposable, in spite of event failures, i.e., $A_S \cong \mathop
{||}\limits_{i = 1}^n F(P_i \left( {A_S } \right))$ if and only if
\begin{itemize}
\item $EF1$: $\forall e_1,
e_2 \in E, q\in Q$: $[\delta(q,e_1)!\wedge \delta(q,e_2)!]\\
\Rightarrow [\exists E_i\in\{E_1, \cdots, E_n\}, \{e_1,
e_2\}\subseteq E_i\backslash \bar{E}_i]\vee[\delta(q, e_1e_2)!
\wedge \delta(q, e_2e_1)!]$;
\item $EF2$: $\forall e_1, e_2 \in E,  q\in Q$, $s\in E^*$: $[\delta(q,
e_1e_2s)!\vee \delta(q, e_2e_1s)!]\\ \Rightarrow [\exists
E_i\in\{E_1, \cdots, E_n\}, \{e_1, e_2\}\subseteq E_i\backslash
\bar{E}_i]\vee [ \delta(q, e_1e_2s)!\wedge \delta(q, e_2e_1s)!]$;
\item $EF3$: $\delta(q_0, \mathop |\limits_{i = 1}^n p_i \left( {s_i }
\right))!$, $\forall \{s_1, \cdots, s_n\}\in \hat{L}\left( {A_S }
\right)$, $\exists s_i, s_j \in \{s_1, \cdots, s_n\}$, $s_i \neq
s_j$, where, $ \hat{L}\left( {A_S } \right) \subseteq L\left( {A_S }
\right)$ is the largest subset of $L\left( {A_S } \right)$ such that
$\forall s\in \hat{L}\left( {A_S } \right), \exists s^{\prime} \in
\hat{L}\left( {A_S } \right),\;\exists \Sigma_i$, $\Sigma_j  \in
\left\{ {\Sigma_1 ,...,\Sigma_n } \right\},
 i \ne j,p_{\Sigma_i  \cap \Sigma_j } \left( s \right)$ and
 $p_{\Sigma_i  \cap \Sigma_j } \left( s^{\prime} \right)$  start with the same
 event, and
\item $EF4$: $\forall i\in\{1, \ldots, n\}$, $x, x_1, x_2 \in Q_i$,
$x_1\neq x_2$, $e\in E_i\backslash \bar{E}_i$, $t_1, t_2 \in
\bar{E}_i^*$, $\delta_i (x, t_1e)= x_1$, $\delta_i (x, t_2e)= x_2$:
$\delta_i (x_1, t_1^{\prime})! \Leftrightarrow \delta_i(x_2,
t_2^{\prime})!$, for some $t_1^{\prime}$, $t_2^{\prime}$ such that
$p_{E_i\backslash \bar{E}_i}(t_1^{\prime}) = p_{E_i\backslash
\bar{E}_i}(t_2^{\prime})$.
\end{itemize}
\end{lemma}

$EF1$-$EF4$ are respectively the decomposability conditions
$DC1$-$DC4$, after event failures with respect to parallel
composition and natural projections into refined local event sets
$\Sigma_i = E_i \backslash \bar{E}_i$, $i\in\{1, \ldots, n\}$,
provided passivity of $\bar{E}_i$, $i\in\{1, \ldots, n\}$.

In this paper we are interested in the case that a task automaton is
not decomposable and would like to ask whether it is possible to
make it decomposable, and if so, whether the automaton can be made
decomposable with minimum number of communication links.
 This problem is
 formally stated as
\begin{problem}\label{Design Problem}
Consider a deterministic task automaton $A_S$ with event set $E =
\mathop {\cup}\limits_{i = 1}^n E_i$ for $n$ agents with local event
sets $E_i$, $i = 1,\ldots, n$. If $A_S$ is not decomposable, can we
modify the sets of private and shared events between local event
sets such that $A_S$ becomes decomposable with respect to parallel
composition and natural projections $P_i$, with the minimum number
of communication links?
\end{problem}

One trivial way to make an automaton $A$ decomposable, is to share
all events among all agents, i.e., $E_i = E$, $\forall i = 1,
\ldots, n$.
This method
, however, is equivalent to centralized control.
In general, in distributed large scale systems, one of the
objectives is to sustain the systems functionalities over as few
number of communication links as possible, as will be addressed in
the next section.

\section{TASK AUTOMATON DECOMPOSABILIZATION}\label{TASK AUTOMATON DECOMPOSABILIZATION}
\subsection{Motivating Examples}
This section is devoted to Problem \ref{Design Problem} and proposes
an approach to redefine the set of private and shared events among
agents in order to make an undecomposable task automaton
decomposable.
For more elaboration, let us to start with a motivating examples.

\begin{example} \label{Belt Conveyors Example}
Consider two sequential belt conveyors feeding a bin, as depicted in
Figure \ref{Belt}.
 To avoid the overaccumulation of materials on
Belt B, when the bin needs to be charged, at first Belt B and then
(after a few seconds), Belt A should be started. After filling the
bin, to stop the charge, first Belt A and then after a few seconds
Belt B is stopped to get completely emptied. The global task
automaton, showing the order of events in this plant, is shown in
Figure \ref{Belts Task Automaton}.
\begin{figure}[ihtp]
      \begin{center}
     \includegraphics[width=0.4\textwidth]{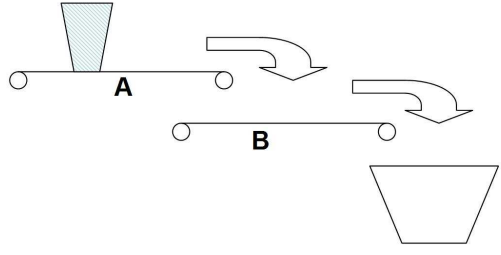}
        \caption{The process of two belt conveyors charging a bin.}
         \label{Belt}
        \end{center}
      \end{figure}
\begin{figure}[ihtp]
      \begin{center}
     $A_S$:  \xymatrix@C=0.5cm{
     \ar[r]&  \bullet \ar[r]^{B_{Start}} &  \bullet \ar[r]_{A_{Start}}& \bullet\ar[r]^{Bin_{Full}}&
     \bullet\ar[r]_{A_{Stop}}&\bullet\ar[r]^{B_{Stop}}&\bullet
\ar`dr_l[lllll]`_u[lllll]^{Bin_{Empty}}[lllll]
   }
        \caption{Global task automaton for belt conveyors and bin.}
 \label{Belts Task Automaton}
        \end{center}
      \end{figure}
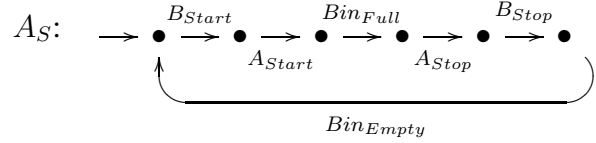

The local event sets for Belt A and Belt B are $E_A=\{A_{Start},
Bin_{Full}, A_{Stop}\}$ and $E_B=\{B_{Start}, B_{Stop},
Bin_{Empty}\}$, respectively, with $A_{Start}$:= Belt A start;
$Bin_{Full}$:= Bin full; $A_{Stop}$:= Belt A stop and wait for 10
Seconds; $B_{Start}$:= Belt B start and wait for 10 Seconds;
$B_{Stop}$:= Belt B stop, and $Bin_{Empty}$: Bin empty.

The task automaton is not decomposable with respect to parallel
composition and natural projection $P_i$, $i\in \{A,B\}$, due to
violation of DC by successive private event pairs $\{B_{Start},
A_{Start} \}$ and $\{A_{Stop}, B_{Stop}\}$. To make $A_S$
decomposable, $(B_{Start}\vee A_{Start} )\wedge(A_{Stop}\vee
B_{Stop})$ should become common between $E_A$ and $E_B$. Therefore,
four options are possible: $(B_{Start} \wedge B_{Stop})$,
$(B_{Start}\wedge A_{Stop})$, $(A_{Start} \wedge B_{Stop})$, or
$(A_{Start} \wedge A_{Stop})$ become common. In each of these
options two private events should become common, and hence, all four
options are equivalent in the sense of optimality. Consider for
example $A_{Start}$ and $A_{Stop}$ to become common. In this case
the new local event sets are formed as $E_A=\{A_{Start}, Bin_{Full},
A_{Stop}\}$ and $E_B=\{B_{Start}, B_{Stop}, Bin_{Empty}, A_{Start},
A_{Stop}\}$. The automaton $A_S$ will then become decomposable
(i.e., $P_A(A_S)||P_B(A_S)\cong A_S$) with the new local event sets
with the corresponding local task automata as are shown in Figure
\ref{Belts Task Automaton Decomposition 1}.
\begin{figure}[ihtp]
      \begin{center}
     $P_A(A_S)$: \xymatrix@C=0.5cm{
     \ar[r]&  \bullet \ar[rr]^{A_{Start}}&&
     \bullet\ar[rr]^{Bin_{Full}}&&\bullet
    \ar`dr_l[llll]`_u[llll]^{A_{Stop}}[llll]
   },\   \   $P_B(A_S)$:\xymatrix@C=0.5cm{
     \ar[r]&  \bullet \ar[r]^{B_{Start}} &  \bullet \ar[r]_{A_{Start}}&
     \bullet\ar[r]^{A_{Stop}}&\bullet\ar[r]_{B_{Stop}}&\bullet
\ar`dr_l[llll]`_u[llll]^{Bin_{Empty}}[llll]
   }
           \caption{Local task automata for belt conveyors, with $E_A=\{A_{Start}, Bin_{Full},
A_{Stop}\}$ and $E_B=\{B_{Start}, B_{Stop}, Bin_{Empty}, A_{Start},
A_{Stop}\}$.}
 \label{Belts Task Automaton
Decomposition 1}
        \end{center}
      \end{figure}
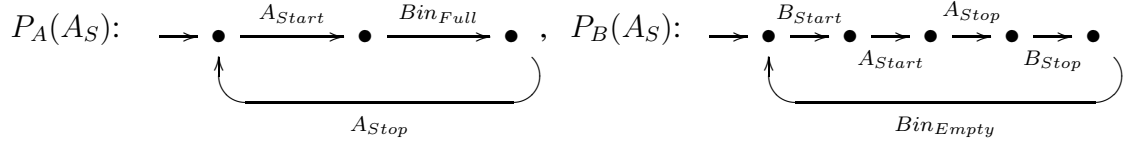
\end{example}

In this example, different sets of private events can be chosen to
make $A_S$ decomposable. All of these sets have the same
cardinality, and hence, no optimality is arisen in this example.
Next example shows a case with different choices of private event
sets to be shared, suggesting optimal decomposition by choosing the
set with the minimum cardinality.
\begin{example}\label{Optimality Concept}
Consider two local event sets $E_1=\{e_1, e_3\}$ and $E_2=\{e_2\}$,
with the global task automaton \xymatrix@R=0.1cm{
                \ar[r]&  \bullet \ar[r]^{e_2} \ar[dr]_{e_1}&\bullet \ar[r]^{e_3}&
     \bullet \\
             && \bullet}.
This automaton is undecomposable due to violation of DC by $e_2 \in
E_2\backslash E_1$ and $\{e_1, e_3\}\in E_1\backslash E_2$. To make
it decomposable, one event among the set $\{e_1, e_2\}$ and another
event among the set $\{e_2, e_3\}$ (either $\{e_2\}$ or $\{e_1,
e_3\}$) should become common. Therefore, in order for optimal
decomposabilization, $\{e_2\}$ is chosen to become common due to its
minimum cardinality. It is obvious that in this case only one event
should become common while if $\{e_1, e_3\}$ was chosen, then two
events were required to be shared.
\end{example}

Motivated by these examples, the core idea in our
decompozabilization approach is to first check the decomposability
of a given task automaton $A_S$, by Lemma \ref{Decomposability
Corollary for n agents}, and if it is not decomposable, i.e., either
of $DC1$-$DC4$ is violated then the proposed method is intended to
make $A_S$ decomposable, by eradicating the reasons of dissatisfying
of decomposability conditions. We will show that violation of
decomposability conditions, can be rooted from two different
sources: it can be because of over-communication among agents, that
may lead to violation of $DC3$ or/and $DC4$, or due to lack of
communication, that may lead to violation of $DC1$, $DC2$, $DC3$
or/and $DC4$. Accordingly, decomposability can be enforced using two
methods of link deletion and link addition, subjected to the type of
undecomposability. Considering link deletion as an intentional event
failure, according to Lemma \ref{Decomposability under event
failure-Theorem} a link can be deleted only if it is passive and its
deletion respects $EF1$-$E4$. On the other hand, the second method
of enforcing of decomposability, i.e., establishing new
communication links, may result in new violations of $DC3$ or $DC4$,
that should be treated, subsequently.

In order to proceed the approach, we firstly introduce four basic
definitions to \emph{detect} the components that contribute in
violation of each decomposability condition and then propose basic
lemmas through which the communication links, and hence the local
event sets are modified to \emph{resolve} the violations of
decomposability conditions.

\subsection{Enforcing $DC1$ and $DC2$}
This part deals with enforcing of $DC1$ and $DC2$. For this purpose,
the set of events that violate $DC1$ or $DC2$ is defined as follows.
%
%

\begin{definition}($DC1\&2$-Violating set)\label{DC1&2-Violating
set} Consider the global task automaton $A_S$ with local event sets
$E_i$ for $n$ agents such that
$E=\overset{n}{\underset{i=1}{\cup}}E_i$. Then, the
$DC1\&2$-Violating set operator $V: A_S\to E \times E$, indicates
the set of event pairs that violate $DC1$ or $DC2$ (violating
pairs), and is defined as $V(A_S):=\{\{e_1, e_2\}| e_1, e_2\in E,
\forall E_i\in\{E_1, \ldots, E_n\}, \{e_1, e_2\}\not\subset E_i,
\exists q\in Q$ such that $\delta(q,e_1)!\wedge \delta(q,e_2)!\wedge
\neg[\delta(q, e_1e_2)!\wedge \delta(q, e_2e_1)!]$ or
$\neg[\delta(q, e_1e_2s)!\Leftrightarrow \delta(q, e_2e_1s)!]\}$,
for some $s\in E^*$. Moreover, $W: A_S\to E$ is defined as
$W(A_S):=\{e\in E| \exists e^{\prime}\in E$ such that $\{e,
e^{\prime}\}\in V(A_S)\}$, and shows the set of events that
contribute in $V(A_S)$ (violating events).
For a particular event $e$ and a specific local event set $E_i \in
\{E_1, \ldots, E_n\}$, $W_e(A_S, E_i)$ is defined as $W_e(A_S, E_i)
= \{e^{\prime} \in E_i|\{e, e^{\prime}\}\in V(A_S)\}$. This set
captures the collection of events from $E_i$ that pair up with $e$
to contribute in violation of $DC1$ or $DC2$. The cardinality of
this set will serve as an index for optimal addition of
communication links to make $V(A_S)$ empty.
\end{definition}

This definition suggests a way to remove a pair of events $\{e_1,
e_2\}$ from $V(A_S)$, by sharing $e_1$ with one of the agents in
$loc(e_2)$ or by sharing $e_2$ with one of the agents in $loc(e_1)$.
Once there exist an agent that knows both event, $loc(e_1)\cap
loc(e_2)$ becomes nonempty and $e_1$ and $e_2$ no longer contribute
in violation of $DC1$ or $DC2$ since $[\exists E_i\in\{E_1, \ldots,
E_n\}, \{e_1, e_2\}\subseteq E_i]$ becomes true for $e_1$ and $e_2$
in Lemma \ref{Decomposability Corollary for n agents}. Therefore,
\begin{lemma}\label{Resolution of DC1 and DC2 violation}
The set $V(A_S)$ becomes empty, if for any $\{e, e^{\prime}\}\in
V(A_S)$, $e$ is included in $E_i$ for some $i\in loc(e^{\prime})$,
or $e^{\prime}$ is included in $E_j$ for some $j\in loc(e)$. In this
case, $\{e, E_i\}$ or $\{e^{\prime}, E_j\}$ is called a $DC1\&
2$-enforcing pair for $DC1\& 2$-violating pair $\{e, e^{\prime}\}$.
\end{lemma}

\begin{example}\label{Resolution of DC1 and DC2-Example }
In Example \ref{Optimality Concept}, $V(A_S) = \{\{e_1, e_2\},
\{e_2, e_3\}\}$, $W(A_S) = \{e_1, e_2, e_3\}$. Including $e_2$ in
$E_1$ vanishes $V(A_S)$ and makes $A_S$ decomposable.
\end{example}

However, applying Lemma \ref{Resolution of DC1 and DC2 violation}
may offer different options for event sharing, since pairs in
$V(A_S)$ may share some events. In this case, the minimum number of
event conversions would be obtained by forming a set of events that
are most frequently shared between the violating pairs. This gives
the minimum cardinality for the set of private events to be shared,
leading to minimum number of added communication links. Such choice
of events offers a set of events that span all violating pairs.
These pairs are captured by $W_e(A_S, E_i)$ for any event $e$. In
order to minimize the number of added communication links for
vanishing $V(A_S)$, one needs to maximize the number of deletions of
pairs from $V(A_S)$ per any link addition. For this purpose, for any
event $e$, $W_e(A_S, E_i)$ is formed to understand the frequency of
appearance of $e$ in $V(A_S)$ for any $E_i$, and then, the event set
$E_i$ with maximum $|W_e(A_S, E_i)|$ is chosen to include $e$ (Here,
$|.|$ denotes the set's cardinality). In this case, inclusion of $e$
in $E_i$ will delete as many pairs as possible from $V(A_S)$.

Interestingly, these operators can be represented using graph theory
as follows. A graph $G=(W,\Sigma)$ consists of a node set $W$ and an
edge set $\Sigma$, where an edge is an unordered pair of distinct
vertices. Two nodes are said to be adjacent if they are connected
through an edge, and an edge is said to be incident to a node if
they are connected. The valency of a node is then defined as the
number of its incident edges \cite{Godsil2001}. Now, since we are
interested in removing the violating pairs by making one of their
events to be shared, it is possible to consider the violating events
as nodes of a graph such that two nodes are adjacent in this graph
when they form a violating pair. This graph is formally defined as
follows.

\begin{definition}($DC1\&2$-Violating Graph)\label{DC1,2-Violating Graph}
Consider a deterministic automaton $A_S$. The $DC1\&2$-Violating
graph, corresponding to $V(A_S)$, is a graph $G(A_S)=(W(A_S),
\Sigma)$. Two nodes $e_1$ and $e_2$ are adjacent in this graph when
$\{e_1, e_2\}\in V(A_S)$.
\end{definition}
In this formulation, the valency of each node $e$ with respect to a
local event set $E_i\in \{E_1, \ldots, E_n\}$ is determined by
$val(e, E_i)=|W_e(A_S, E_i)|$. When $e$ is included into $E_i$, it
means that all violating pairs containing $e$ and events from $E_i$
are removed from $V(A_S)$, and equivalently, all corresponding
incident edges are removed from $G(A_S)$. For this purpose,
following algorithm finds the set with the minimum number of private
events to be shared, in order to satisfy $DC1$ and $DC2$. The
algorithm is accomplished on graph $G(A_S)$, by finding $e$ and
$E_i$ with maximum $|W_e(A_S, E_i)|$ and including $e$ in $E_i$,
deleting all edges from $e$ to $E_i$, updating $W(A_S)$, and
continuing until there is not more edges in $G(A_S)$ to be deleted.

\begin{algorithm}
\label{minimum spanning nodes-n agents} \ \begin{enumerate}
\item For a deterministic automaton $A_S$, with local event sets
$E_i$, $i = 1, \ldots, n$, violating $DC1$ or $DC2$, form the
$DC1\&2$-Violating graph ; set $E_i^0 = E_i$, $i = 1, \ldots, n$;
$V^0(A_S) = V(A_S)$; $W^0(A_S) = W(A_S)$; $G^0(A_S)=(W(A_S),
\Sigma)$; k=1;
\item Among all events in the nodes in $W^{k-1}(A_S)$, find $e$ with
the maximum $|W_e^{k-1}(A_S, E_i^{k-1})|$, for all $E_i^{k-1}\in
\{E_1^{k-1}, \ldots, E_n^{k-1}\}$;
\item $E_i^k = E_i^{k-1} \cup \{e\}$;  and delete all edges from $e$ to $E_i^k$;
\item update $W_e^k(A_S, E_i)$ for all nodes of $G(A_S)$;
\item set $k=k+1$ and go to step $(2)$;
\item continue, until there exist no edges.
\end{enumerate}

\end{algorithm}

This algorithm successfully terminates due to finite set of edges
and nodes in the graph $G(A_S)$ and enforces $A_S$ to satisfy $DC1$
and $DC2$ as

\begin{lemma}\label{minimum spanning nodes-n agents-lemma}
Algorithm \ref{minimum spanning nodes-n agents} leads $A_S$ to
satisfy $DC1$ and $DC2$ with minimum addition of communication
links. Moreover if $A_S$ satisfies $DC3$ and $DC4$ and $E_i^k =
E_i^{k-1} \cup \{e\}$ in Step $3$ does not violate $DC3$ and $DC4$
in all iterations, then Algorithm \ref{minimum spanning nodes-n
agents} makes $A_S$ decomposable with minimum addition of
communication links.
\end{lemma}

\begin{proof}
See the Appendix for proof.
\end{proof}

\begin{remark}(Special case: Two agents)\label{valency}
For the case of two agents, since they are only two local event
sets, for all $\{e, e^{\prime}\}\in V(A_S)$, $e$ and $e^{\prime}$
are from different local event sets, and hence, for $n = 2$,
$|W_e(A_S, E_i)|$ is equivalent to $val(e)$, and addition of $e$
into $E_i$ in each step implies the deletion of all incident edges
of $e$.
\end{remark}

\begin{remark}\label{Resolution of DC1 and DC2 may violate DC3 or DC4}
Although Algorithm \ref{minimum spanning nodes-n agents} leads $A_S$
to satisfy $DC1$ and $DC2$, it may cause new violations of $DC3$
or/and $DC4$, due to establishing new communication links.
\end{remark}

\begin{example}\label{Example-Resolution of DC1 and DC2 may violate DC3 or DC4}
Consider a task automaton $A_S$:\\ \xymatrix@R=0.1cm{
               \bullet &\bullet \ar[l]_{e_5} &\bullet\ar[l]_{e_1}&
                \ar[r] & \bullet\ar[dr]_{e_2}\ar[dl]^{a}
                \ar[r]^{e_1}& \bullet
               \ar[r]^{e_2} & \bullet\\
               & \bullet & \bullet \ar[l]^b& \bullet
               \ar[l]^{e_1}\ar[ul]_{e_3} && \bullet \ar[r]_{e_4}& \bullet
               \ar[r]_{e_6}&\bullet} with local event sets $E_1 = \{a, b, e_1,\\ e_3,
               e_5\}$ and $E_2 = \{a, b, e_2, e_4, e_6\}$. Both
               $DC1$ and $DC2$ are violated by event pair $\{e_1,
               e_2\}$ when they require decision on a choice and a
               decision on their order from the initial state,
               while none of the agents knows both of them. To
               vanish $V(A_S) = \{\{e_1, e_2\}\}$, two enforcing pairs
               are suggested: $\{e_1, E_2\}$ ($e_1$ to be included in
               $E_2$) or $\{e_2, E_1\}$ ($e_2$ to be included in
               $E_1$). However, inclusion of $e_1$ in $E_2$, cause a
               new violation of $DC4$ since with new $E_2 = \{a, b,
               e_1, e_2, e_4, e_6\}$, $P_2(A_S)$ is obtained as
               $P_2(A_S)$: \xymatrix@R=0.1cm{
               & \bullet& \ar[r]&\bullet \ar[dl]^{a} \ar[r]^{e_1}& \bullet
               \\
            \bullet & \bullet \ar[l]^b& \bullet
               \ar[l]^{e_1}\ar[ul]_{e_1}}, violating $DC4$, due to
               new nondeterminism, for which $e_3$ also is required
               to be included to $E_2$ in order to make $A_S$
               decomposable. On the other hand, if instead of
               including $e_1$ in $E_2$, one included $e_2$ in
               $E_1$, then besides violation of $DC4$ (as there does not
               exists a deterministic automaton that bisimulates
               $P_2(A_S)$), new violations of $DC3$ emerged, as with
               new event set $E_1 = \{a, b, e_1, e_2, e_3, e_5\}$, the
               parallel composition of $P_1(A_S)$: \xymatrix@R=0.1cm{
               \bullet &\bullet \ar[l]_{e_5} &\bullet\ar[l]_{e_1}&
                \ar[r] & \bullet\ar[dr]_{e_2}\ar[dl]^{a}
                \ar[r]^{e_1}& \bullet
               \ar[r]^{e_2} & \bullet\\
               & \bullet & \bullet \ar[l]^b& \bullet
               \ar[l]^{e_1}\ar[ul]_{e_3} && \bullet} and
               $P_2(A_S)$:\\
               \xymatrix@R=0.1cm{
               &&  \ar[r] & \bullet\ar[dr]_{e_2}\ar[dl]^{a}
                \ar[r]^{e_2}& \bullet\\
               & \bullet & \bullet \ar[l]_b
             && \bullet \ar[r]^{e_4}& \bullet
               \ar[r]^{e_6}&\bullet} produces string $e_1e_2e_4e_6$
               that does not appear in $A_S$. To make $A_S$ decomposable, we also need
               to include $e_1$ and $e_3$ in $E_2$.
\end{example}

\subsection{Enforcing $DC3$} Lemma \ref{Resolution of DC1 and DC2
violation} proposes adding communication links to make $DC1$ and
$DC2$ satisfied. Next step is to deal with violations of $DC3$. In
contrast to the cases for $DC1$ and $DC2$, violation of $DC3$ can be
overcome either by disconnecting one of its communication links to
prevent the illegal synchronization of strings, or by introducing
new shared events to fix strings and avoid illegal interleavings.

To handle violation of $DC3$, we firstly define the set of tuples
that violate $DC3$ as follows.

\begin{definition}($DC3-violating$ tuples)\label{DC3-violating Events}
Consider a deterministic automaton $A_S$, satisfying $DC1$ and $DC2$
 and let $\tilde L\left(
{A_S } \right) \subseteq L\left( {A_S } \right)$ be the largest
subset of $L\left( {A_S } \right)$ such that $\forall s\in \tilde
L\left( {A_S } \right)\exists s^{\prime} \in \tilde L\left( {A_S }
\right),\;\exists E_i, E_j  \in \left\{ {E_1 ,...,E_n } \right\},
 i \ne j,p_{E_i  \cap E_j } \left( s \right)$ and
  $p_{E_i  \cap E_j } \left( s^{\prime} \right)$
start with the same event $a\in E_i \cap E_j$. For any such $E_i$,
$E_j$ and $a$,
 if $\exists \{s_1, \cdots, s_n\}\in L\left( {A_S } \right)$,
$\exists s_i, s_j \in \{s_1, \cdots, s_n\}, s_i \neq s_j$, $s_i, s_j
\in \tilde L\left( {A_S } \right)$, $\neg\delta(q_0, \mathop
|\limits_{i = 1}^n p_i \left( {s_i } \right))!$, then
 $a$ is called a $DC3-violating$ event with respect to $s_1$, $s_2$, $E_i$ and
$E_j$, and  $(s_1, s_2, a,  E_i, E_j)$ is called a $DC3$-violating
tuple. The set of all $DC3-violating$ tuples is denoted by $DC3-V$
and defined as $DC3-V = \left\{(s_1, s_2, a,  E_i, E_j)|e \right.$
is a DC3-violating event with respect to $s_1$, $s_2$, $E_i$ and
$E_j\left.\right\}$.
\end{definition}

Any violation in $DC3$ can be interpreted in two ways: firstly, it
can be seen as over-communication of shared event $a$ that lead to
synchronization of $s_1$ and $s_2$ in $(s_1, s_2, a,  E_i, E_j)$ and
emerging illegal interleaving strings from composition of $P_i(A_S)$
and $P_j(A_S)$. In this case, if event $a$ is excluded from $E_i$ or
$E_j$, then  $a$ will no longer contribute in synchronization to
generate illegal interleavings, and hence, $(s_1, s_2, a,  E_i,
E_j)$ will no longer remain a $DC3$-violating tuple. However,
exclusion of $a$ from $E_i$ or $E_j$ is allowed, only if it is
passive (exclusion is considered as an intentional event failure)
and does not violate $EF1$-$EF4$. The second interpretation reflects
a violation of $DC3$ as a lack of communication, such that if for
any $DC3$ violating tuple $(s_1, s_2, a,  E_i, E_j)$, one event that
appears before $a$ in $s_1$ or $s_2$, is shared between $E_j$ and
$E_j$, then $P_i(A_S)$ and $P_j(A_S)$ will have enough information
to distinguish $s_1$ and $s_2$ to prevent illegal interleaving of
strings.
 Two methods for resolving
the violation of $DC3$ can be therefore stated as the following
lemma.
\begin{lemma}\label{Resolution of DC3 violation}
Consider an automaton $A_S$, satisfying $DC1$ and $DC2$. Then any
$DC3$-violating tuple $(s_1, s_2, a,  E_i, E_j)$ is overcome, when:
\begin{enumerate}
\item $a$ is excluded from $E_i$ or $E_j$ (eligible if it respects passivity and $EF1$-$EF4$), or
\item if $\exists b\in (E_i\cup E_j)\backslash (E_i\cap E_j)$ that appears before $a$
in only one of $s_1$ and $s_2$, then $b$ is included in $E_i\cap
E_j$, otherwise, pick $e_1\in p_{E_i\cup E_j}(s_1)$, $e_2\in
p_{E_i\cup E_j}(s_2)$, such that $e_1 \neq e_2$, $e_1$, $e_2$ appear
before $a$ in $s_1$ and $s_2$, are included in $E_i\cap E_j$.
\end{enumerate}
\end{lemma}

To handle a violation of $DC3$, when, $b\in E_i\backslash E_j$ is to
be included in $E_j$, then $\{b, E_j\}$ is called a $DC3$-enforcing
pair; while, when $\{e_1, e_2\}\subseteq E_i\backslash E_j$ has to
be included in $E_j$, then $\{\{e_1, e_2\}, E_j\}$ is denoted as
$DC3$-enforcing tuple. Finally, when $e_1 \in E_i\backslash E_j$ and
$e_2\in E_j\backslash E_i$ have to be included in $E_j$ and $E_i$,
respectively, then $\{\{e_1, E_j\}, \{e_2, E_i\}\}$ is called a
$DC3$-enforcing tuple.
\begin{proof}
See the proof in the Appendix.
\end{proof}

\begin{remark}\label{Conflicts in Resolution of DC3
violation} Applying the first method in Lemma \ref{Resolution of DC3
violation}, namely, exclusion of $a$ from $E_i$ or $E_j$ in a
$DC3$-violating tuple $(s_1, s_2, a,  E_i, E_j)$, is only allowed if
$a$ is passive in that local event set, and the exclusion does not
violate $EF1$-$EF4$. The reason is that once a shared event $a\in
E_i\cap E_j$ becomes a private one in for example $E_i$, then
decision makings on the order/selection between any $e\in
E_i\backslash a$ and $a$ cannot be accomplished by the $i-th$ agent,
and if there is no other agent to do so, then $A_S$ becomes
undecomposable. Moreover, deletion of a communication link may also
result in generation of new interleavings in the composition of
local automata, that are not legal in $A_S$ (violation of $EF3$). In
addition, deletion of $a$ from $E_i$ may impose a nondeterminism in
bisimulation quotient of $P_i(A_S)$, leading to violation of $EF4$.
On the other hand, the second method, namely, establishing new
communication link by sharing $b$ with $E_i$ or $E_j$ may lead to
new violations of $DC3$ or $DC4$ that have to be avoided or
resolved, subsequently.

Both methods in Lemmas \ref{Resolution of DC3 violation} present
ways to resolve the violation of $DC3$. They differ however in the
number of added communication links, as the first method deletes
links, whereas the second approach adds communication links to
enforce $DC3$. Therefore, in order to have as few number of links as
possible among the agents, one should start with the link deletion
method first, and if it is not successful due to violation of
passivity or any of $EF1$-$EF4$, then link addition is used to
remove $DC3$-violating tuples from $DC3-V$.
\end{remark}

\begin{example}\label{Conflicting DC3 with EF1, EF2 and EF4}
This example shows an undecomposable automaton that suffers from a
conflict on a communication link whose existence violates $DC3$,
whereas its deletion dissatisfies $EF1$, $EF2$ and $EF4$.

Let $snd_e(i)$ and $rcv_e(i)$ respectively denote the set of labels
that $A_i$ sends $e$ to those agents and the set of labels that
$A_i$ receives $e$ from  their agents, defined as $snd_e(i)=\{j\in
\{1,...,n\}|A_i \hbox{ sends }\\ e \hbox{ to } A_j\}$ and
$rcv_e(i)=\{j\in \{1,...,n\}|i\in snd_e(j)\}$. Consider the task
automaton
$A_S$:\\ \xymatrix@R=0.3cm{
 \bullet & \bullet \ar[l]_{b}&  \bullet \ar[l]_{e_1}
 \ar[r]^{e_5} &  \bullet \ar[r]^{e_1}&  \bullet \\
            \bullet & \bullet \ar[l]_{b} && \ar[d] &  \bullet  \ar[dr]^{e_2}    \\
&& \bullet \ar[ul]_{a} \ar[dl]^{b} &  \bullet \ar[l]_{c}\ar[uul]_{d} \ar[ur]^{e_1} \ar[dr]_{e_2} \ar[d]_{a}& &  \bullet  \ar[r]^{a}& \bullet \ar[r]^{e_3}& \bullet \\
  \bullet & \bullet \ar[l]^{e_2} &\bullet & \bullet \ar[l]^{e_2} & \bullet
  \ar[ur]_{e_1} }
with communication pattern\\ $2 \in snd_{a, b, c, d}(1)$, $1 \notin
snd_{a, b, c, d}(1)$ and local event sets $E_1=\{a, b, c, d, e_1,
e_3, e_5\}$, $E_2=\{a, b, c, d, e_2\}$, leading to $P_1(A_S)$:
\xymatrix@R=0.1cm{& \bullet & \bullet \ar[l]_{b}&  \bullet
\ar[l]_{e_1}
 \ar[r]^{e_5} &  \bullet \ar[r]^{e_1}&  \bullet \\
\bullet & \bullet \ar[l]_{b} & \ar[dr]&&  \bullet  \ar[r]^{a}&  \bullet  \ar[r]^{e_3}& \bullet   \\
& \bullet & \bullet\ar[l]^{b}\ar[ul]_{a} &  \bullet
\ar[l]^{c}\ar[uu]^{d}\ar[ur]^{e_1}\ar[r]_{a}& \bullet
                },\\
$P_2(A_S)$: \xymatrix@R=0.1cm{ \bullet & \bullet \ar[l]_{b} &\bullet &\ar[d]&  \bullet  \ar[r]^{a}&  \bullet     \\
\bullet & \bullet \ar[l]^{e_2} & \bullet\ar[l]^{b}\ar[ul]_{a} &
\bullet \ar[l]^{c}\ar[ul]_{d}\ar[ur]^{e_2}\ar[r]_{a}& \bullet
\ar[r]_{e_2}& \bullet
                } and
$P_1(A_S)||P_2(A_S)$: \\\xymatrix@R=0.3cm{ &&&&&&\bullet
\ar[dr]^{e_2}\\
\bullet & \bullet \ar[l]_{b}&  \bullet \ar[l]_{e_1}
 \ar[r]^{e_5} &  \bullet \ar[r]^{e_1}&  \bullet &\bullet \ar[ur]^{e_3}\ar[dr]_{e_2}&&\bullet\\
            \bullet & \bullet \ar[l]_{b} && \ar[d] &  \bullet  \ar[dr]^{e_2}\ar[ur]^{a}&&\bullet\ar[ur]_{e_3}    \\
&& \bullet \ar[ul]_{a} \ar[dl]^{b} &  \bullet \ar[uul]_{d} \ar[l]_{c} \ar[ur]^{e_1} \ar[dr]_{e_2} \ar[d]_{a}& &  \bullet  \ar[r]^{a}& \bullet \ar[r]^{e_3}& \bullet \\
  \bullet & \bullet \ar[l]^{e_2} &\bullet & \bullet \ar[l]^{e_2} & \bullet
  \ar[ur]_{e_1}}
which is not bisimilar to $A_S$. Here, $A_S$ is not decomposable
since two strings $e_1ae_2e_3$ and $e_1ae_3e_2$ are newly generated
from the interleaving of strings in $P_1(A_S)$ and $P_2(A_S)$, while
they do not appear in $A_S$, and hence, $DC3$ is not fulfilled, due
to $DC3$-violating tuples $(e_1e_2ae_3, ae_2, a, E_1, E_2)$ and
$(e_2e_1ae_3, ae_2, a, E_1, E_2)$. Now, as Lemma \ref{Resolution of
DC3 violation}, one way to fix the violation of $DC3$ is by
excluding $a$ from $E_2$. However, although $a$ is passive in $E_2$,
its exclusion from $E_2$ dissatisfies $EF1$( as $\delta(q_0, e_2)!
\wedge \delta(q_0, a)! \wedge \neg [\delta(q_0, e_2a)!\wedge
\delta(q_0, ae_2)!]$) and $EF2$ (since $\delta(q_0, e_1e_2a)! \wedge
\neg \delta(q_0, e_1ae_2)!$). In this case, $DC4$ also will be
violated as $P_2(A_S)$ becomes $P_2(A_S)\cong \xymatrix@R=0.1cm{
\ar[r]&\bullet
\ar[dr]_{e_2}\ar[dl]^{d}\ar[r]^{c}&\bullet \ar[r]^{b}\ar[dr]_{b}&\bullet\\
\bullet &&\bullet &\bullet \ar[r]^{e_2}& \bullet}$ that bisimulates
no deterministic automaton.

Lemma \ref{Resolution of DC3 violation} also suggests another method
to enforce $DC3$, by including either $e_1$ in $E_2$ or $e_2$ in
$E_1$. Inclusion of $e_1$ in $E_2$, however, leads to another
violation of $DC4$, as it produces a nondeterminism after event $d$.
This in turn will need to include $e_5$ in $E_2$ to make $A_S$
decomposable. Alternatively, instead of inclusion of $e_1$ in $E_2$,
one can include $e_2$ in $E_1$, that enforces $DC3$ and makes $A_S$
decomposable. The second method of Lemma \ref{Resolution of DC3
violation} is more elaborated in the next example.
\end{example}

\begin{example}\label{enforcing DC3 using link addition}
This example shows handling of $DC3$-violating tuples using the
second method in Lemma \ref{Resolution of DC3 violation}, i.e., by
event sharing. Later on, this example will be also used to
illustrate the enforcement of $DC4$. Now, consider a task automaton
$A_S$: \xymatrix@C=0.4cm{ \bullet \ar[r]^{e_3} &\bullet\ar[r]^{e_5}
& \bullet \ar[r]^{a}& \bullet \ar[r]^{e_2}& \bullet
 \\
\ar[r]&\bullet \ar[ul]^{e_1}\ar[dl]_{e_5}\ar[r]^{a}&\bullet
\ar[r]^{e_6}&\bullet\\
\bullet \ar[r]_{e_3} &\bullet\ar[r]_{e_1} & \bullet \ar[r]_{a}&
\bullet \ar[r]_{e_4}& \bullet} with local event sets $E_1 = \{a,
e_1, e_3, e_5\}$ and $E_2 = \{a, e_2, e_4, e_6\}$, and let three
branches in $A_S$ from top to bottom to be denoted as $s_1:=
e_1e_3e_5ae_2$, $s_3:= ae_6$ and $s_2:= e_5e_3e_1ae_4$. This
automaton does not satisfy $DC4$ (as $P_2(A_S)$ has no deterministic
bisimilar automaton), as well as $DC3$, as the parallel composition
of $P_1(A_S)$: \xymatrix@C=0.4cm{ \bullet \ar[r]^{e_3}
&\bullet\ar[r]^{e_5} & \bullet \ar[r]^{a}& \bullet
 \\
\ar[r]&\bullet \ar[ul]^{e_1}\ar[dl]_{e_5}\ar[r]^{a}&\bullet
\\
\bullet \ar[r]_{e_3} &\bullet\ar[r]_{e_1} & \bullet \ar[r]_{a}&
\bullet} and $P_2(A_S)$: \xymatrix@C=0.4cm{ \bullet \ar[r]^{e_2}
&\bullet
 \\
\ar[r]&\bullet \ar[ul]^{a}\ar[dl]_{a}\ar[r]^{a}&\bullet
\ar[r]^{e_4}&\bullet\\
\bullet \ar[r]_{e_6} &\bullet} have illegal interleaving strings
$\{e_1e_3e_5ae_6, e_5e_3e_1ae_2\}$, $e_1e_3e_5ae_4$ and
$e_5e_3e_1ae_4$, corresponding to $DC3$-violating tuples $(s_1, s_2,
a, E_1, E_2)$, $(s_1, s_3, a, E_1, E_2)$ and $(s_2, s_3$, $a, E_1,
E_2)$, respectively.

For pairs of strings $\{s_1, s_3\}$ and $\{s_2, s_3\}$, there exits
an event $e_5\in (E_1\cup E_2)\backslash (E_1\cap E_2)$ that appears
before $a$, only in $s_1$ and $s_2$, but not in $s_3$. Therefore,
inclusion of $e_5$ in $E_2$, removes the illegal interleavings
between $s_1$ and $s_2$ with $s_3$, but not across $s_1$ and $s_2$,
as with new $E_2 =\{a, e_2, e_4, e_5, e_6\}$ and $P_2(A_S)$:
\xymatrix@C=0.4cm{ \bullet\ar[r]^{a} & \bullet \ar[r]^{e_2}& \bullet
 \\
\ar[r]&\bullet \ar[ul]^{e_5}\ar[dl]_{e_5}\ar[r]^{a}&\bullet
\ar[r]^{e_6}&\bullet\\
\bullet \ar[r]_{a}& \bullet \ar[r]_{e_4}& \bullet},  $(s_1, s_3, a,
E_1, E_2)$ and $(s_2, s_3, a, E_1, E_2)$ are no longer
$DC3$-violating tuples, while $(s_1, s_2, a, E_1, E_2)$ still
remains a $DC3$-violating one with illegal interleavings
$e_1e_3e_5ae_4$ and $e_5e_3e_1ae_2$. The reason is that $e_5$
appears before $a$ in both $s_1$ and $s_2$, and there is no event
that appear before $a$ only in one of the strings $s_1$ and $s_2$.
For this case, according to Lemma \ref{Resolution of DC3 violation},
two different events that appear before ``a'', one from $p_{E_1\cup
E_2}(s_1) = s_1$ and the other from $p_{E_1\cup E_2}(s_2) = s_2$,
i.e., $e_1$ and $e_5$ have to be attached to $E_2$, resulting in
$E_2 = \{a, e_1, e_2, e_4, e_5, e_6\}$, \xymatrix@C=0.4cm{
\bullet\ar[r]^{e_5} & \bullet \ar[r]^{a}& \bullet \ar[r]^{e_2}&
\bullet
 \\
\ar[r]&\bullet \ar[ul]^{e_1}\ar[dl]_{e_5}\ar[r]^{a}&\bullet
\ar[r]^{e_6}&\bullet\\
\bullet\ar[r]_{e_1} & \bullet \ar[r]_{a}& \bullet \ar[r]_{e_4}&
\bullet} and $P_1(A_S)||P_2(A_S) \cong A_S$.
\end{example}

\subsection{Enforcing $DC4$}
Similar to $DC1$-$DC3$, a violation of $DC4$ can be regarded as a
lack of communication link that causes nondeterminism in a local
task automaton. Such interpretation calls for establishing a new
communication link to prevent the emergence of local nondeterminism.
Moreover, when this local nondeterminism occurs on a shared event,
the corresponding violation of $DC4$ can be overcome by excluding
the shared event from the respective local event set. It should be
noted however that the event exclusion should respect the passivity
and $EF1$-$EF4$ conditions. When $DC4$ is enforced by link
additions, similar to what we discussed for $DC3$, addition of new
communication link may cause new violations of $DC3$ or/and $DC4$.
To enforce $DC4$, firstly a $DC4$-violating tuple is defined as
follows.

\begin{definition}($DC4-violating$ tuple)\label{DC4-violating tuple}
Consider a deterministic automaton $A_S$ with local event sets $E_i
= 1, \ldots, n$, $\forall i\in\{1,...,n\}$, $q, q_1, q_2 \in Q$,
$t_1, t_2 \in (E\backslash E_i)^*$, $e\in E_i$, $\delta(q, t_1e) =
q_1 \neq \delta(q, t_2e) = q_2$, $\exists t \in E^*$, $\delta(q_1,
t)!$, but $\nexists t^{\prime} \in E^*$ such that $\delta(q_2,
t^{\prime})!$, $p_i(t) = p_i (t^{\prime})$. Then, $(q, t_1, t_2, e,
E_i)$ is called a $DC4$-violating tuple.
\end{definition}

This definition suggests the way to overcome the violation of $DC4$,
as stated in the following lemma.

\begin{lemma}\label{Resolution of DC4 violation}
Any $DC4$-violating tuple $(q, t_1, t_2, e, E_i)$ is overcome, when:
\begin{enumerate}
\item $e$ is excluded from $E_i$, (eligible, if it is passive in $E_i$ and its
exclusion respects $EF1-EF4$), or
\item if $\exists e^{\prime} \in (t_1\cup t_2)\backslash (t_1\cap
t_2)$, $e^{\prime}$ is included in $E_i$; otherwise, $e_1\in t_1$
and $e_2\in t_2$, such that $e_1 \neq e_2$, are included in $E_i$.
In these cases, $\{e^{\prime}, E_i\}$ and $\{\{e_1, e_2\}, E_i\}$
are called $DC4$-enforcing tuples.
\end{enumerate}
\end{lemma}
\begin{proof}
See the proof in the Appendix.
\end{proof}

Following examples illustrate the methods in Lemma \ref{Resolution
of DC4 violation} to enforce $DC4$.

\begin{example}\label{enforcing DC4 by event inclusion and event
exclusion} This example shows an automaton that is undecomposable
due to a violation in $DC4$, while $DC4$ can be enforced using both
methods: event exclusion as well as event inclusion. Consider the
task automaton $A_S$: \xymatrix@R=0.1cm{
              \ar[r]&  \bullet \ar[r]^{e_1} \ar[dr]_{a}   & \bullet  \ar[r]^{a}  & \bullet\ar[r]^{b}& \bullet \ar[r]^{e_2}& \bullet  \\
             &    &\bullet \ar[r]^{e_3}& \bullet
                } with $E_1=\{a, b, e_1, e_3\}$, $E_2=\{a, b,
                e_2\}$, $2\in snd_{a, b}(1)$, $1\notin snd_{a,
                b}(2)$,
               leading to
                $P_1(A_S)$:\\  \xymatrix@R=0.1cm{
          \ar[r]&  \bullet \ar[r]^{e_1} \ar[dr]_{a}  & \bullet  \ar[r]^{a}  & \bullet\ar[r]^{b}& \bullet  \\
             &    &\bullet \ar[r]^{e_3}& \bullet
                },
                 $P_2(A_S)$:\xymatrix@R=0.1cm{
\ar[r]&  \bullet \ar[r]^{a} \ar[dr]_{a}   & \bullet  \ar[r]^{b}  &  \bullet \ar[r]^{e_2}  &  \bullet \\
             &    & \bullet
                }, and\\
$\mathop {||}\limits_{i = 1}^2 P_i (A_S )\cong \xymatrix@R=0.3cm{
\ar[r]&  \bullet \ar[r]^{e_1} \ar[d]_{a}   & \bullet  \ar[r]^{a} \ar[dr]_{a} & \bullet\ar[r]^{b}& \bullet \ar[r]^{e_2}& \bullet \\
 & \bullet \ar[r]^{e_3}& \bullet &\bullet
                }$ which is not bisimilar to $A_S$, due to
                violation of $DC4$ as there does not exist a
                deterministic automaton $P^{\prime}_2(A_S)$ such
                that $P^{\prime}_2(A_S)\cong P_2(A_S)$. Here, $(q_0, t_1 = e_1, t_2 = \varepsilon, a,
                E_2)$ is a $DC4$-violating tuple. Since $a$ is
                passive in $E_2$ and its exclusion from $E_2$ keeps
                $EF1$-$EF4$ valid, according to Lemma \ref{Resolution
of DC4 violation}, one way to enforce $DC4$ is exclusion of $a$ from
$E_2$, resulting in $E_2 = \{b, e_2\}$, $P_2(A_S)$:
\xymatrix@R=0.1cm{ \ar[r]&  \bullet \ar[r]^{b} & \bullet
\ar[r]^{e_2}  &  \bullet } and $P_1(A_S)||P_2(A_S) \cong A_S$.

Another suggestion of Lemma \ref{Resolution of DC4 violation} to
overcome the $DC4$-violating tuple $(q_0, t_1 = e_1, t_2 =
\varepsilon, a, E_2)$ is addition of a communication link to prevent
the nondeterminism in $P_2(A_S)$. Since there exists $e_1$ that
appears before $a$ in $t_1$ only, inclusion of $e_1$ in $E_2$ also
enforces $DC4$ as with new $E_2 = \{a, b, e_1, e_2\}$, $P_2(A_S)$:
\xymatrix@R=0.1cm{
          \ar[r]&  \bullet \ar[r]^{e_1} \ar[dr]_{a}  & \bullet  \ar[r]^{a}  & \bullet\ar[r]^{b}& \bullet  \\
             &    &\bullet
                } and $\mathop {||}\limits_{i = 1}^2 P_i (A_S )\cong
                A_S$.
For the cases that there does not exist an event $b$ that appears
before $a$ in only one of the strings $t_1$ or $t_2$, according to
Lemma \ref{Resolution of DC4 violation}, one needs to attach one
event from each of two strings $t_1$ and $t_2$ in $E_i$. For
instance consider the $DC4$-violating tuple $(t_1 = e_1e_3e_5, t_2 =
e_5e_3e_1, a, E_2)$ in Example \ref{enforcing DC3 using link
addition}, with no event that appears before $a$ in $(t_1\cup
t_2)\backslash (t_1\cap t_2)$. In that case $\{e_1\in t_1, e_5\in
t_2\}$ can be included in $E_2$ to make $A_S$ decomposable, as it
was shown in Example \ref{enforcing DC3 using link addition}.

\end{example}

\begin{example}\label{enforcing DC4 not by event exclusion}

Example \ref{enforcing DC4 by event inclusion and event exclusion}
showed a violation of $DC4$ that could be overcome using both method
in Lemma \ref{Resolution of DC4 violation}, namely, by link deletion
and link addition. In Example \ref{enforcing DC4 by event inclusion
and event exclusion}, event $a$ was a passive shared event whose
exclusion from $E_2$ respected $EF1$-$EF4$, otherwise it was not
allowed to be excluded. If the task automaton was \xymatrix@R=0.1cm{
              \ar[r]&  \bullet \ar[r]^{e_1} \ar[dr]_{a}   & \bullet  \ar[r]^{a}  & \bullet\ar[r]^{e_2}& \bullet \ar[r]^{b}& \bullet  \\
             &    &\bullet \ar[r]^{e_3}& \bullet
                } with $E_1=\{a, b, e_1, e_3\}$, $E_2=\{a, b,
                e_2\}$, then $DC4$ could not be enforced by exclusion
                of $a$ from $E_2$, as $EF2$ was violated since after
                this exclusion, no agent can handle the decision
                making on the order of $a$ and $e_2$.
Another constraint for link deletion is the passivity of the event.
 For example, consider $A_S^{\prime}$:\xymatrix@R=0.1cm{ &&&\bullet\ar[r]^{e_4}&\bullet\\
                \ar[r]&\bullet \ar[r]^{e_1}\ar[dr]_{e_2}&\bullet   \ar[r]_{e_2}\ar[ur]^{a}&\bullet\\
             && \bullet \ar[r]_{e_1}&\bullet} with $E_1 = \{e_1, a\}$, $E_2
             = \{e_2, e_4, a\}$. $A_S^{\prime}$ is not decomposable due
             to violation of $DC4$ in $P_1(A_S)$: \xymatrix@R=0.1cm{
                \ar[r]&\bullet \ar[r]^{e_1}\ar[dr]_{e_1}&\bullet \ar[r]^{a}&\bullet  \\
             && \bullet }. The nondeterminism in $P_1(A_S)$, and accordingly the $DC4$-violating tuple
             $(q_0, \varepsilon, e_2, e_1, E_1)$, cannot be removed by event exclusion since it occurs on $e_1$
             that is not a shared event. To enforce $DC4$ according
             to Lemma \ref{Resolution of DC4 violation}, $e_2$
             is required to be included into $E_1$ that makes
             $A_S^{\prime}$ decomposable.
\end{example}

Another important issue for addition of communication link to
enforce $DC4$ is that establishing new communication link may lead
to new violations of $DC3$ or $DC4$, as it is shown in the following
example.

\begin{example}\label{Enforcing DC4 leading to new violations}
Assume the task automaton in Example \ref{enforcing DC4 by event
inclusion and event exclusion} had a part as shown in he left hand
side of the initial state in $A_S$:\\ \xymatrix@R=0.1cm{ \bullet
&\bullet \ar[l]_{d} &\bullet \ar[l]_{e_1} &\bullet
\ar[l]_{e_5}\ar[dl]^{e_1} &\bullet
\ar[l]_{c}\ar[r]^{e_1}\ar[dr]_{a}& \bullet \ar[r]^{a} & \bullet
\ar[r]^{b} & \bullet
\ar[r]^{e_2} & \bullet\\
&&\bullet &&\ar[u]& \bullet \ar[r]_{e_3}& \bullet} with $E_1 = \{a,
b, c$, $d, e_1, e_3,\\ e_5\}$, $E_2
             = \{a, b, c, d, e_2\}$. Identical to Example \ref{enforcing DC4 by event
inclusion and event exclusion}, $(q_0, t_1 = e_1, t_2 = \varepsilon,
a, E_2)$ is a $DC4$-violating tuple and can be overcome by excluding
$a$ from $E_2$, removing the nondeterminism on $a$ in $P_2(A_S)$.
However, unlike Example \ref{enforcing DC4 by event inclusion and
event exclusion}, including $e_1$ into $E_2$ (i.e., $E_2 = \{a, b,
c, d, e_1, e_2\}$), leads to a new violation of $DC4$ in $P_2(A_S)$:
\xymatrix@R=0.1cm{ \bullet &\bullet \ar[l]_{d} &\bullet \ar[l]_{e_1}
\ar[dl]^{e_1} &\bullet \ar[l]_{c}\ar[r]^{e_1}\ar[dr]_{a}& \bullet
\ar[r]^{a} & \bullet
\ar[r]^{b} & \bullet \ar[r]^{e_2} & \bullet\\
&\bullet &&\ar[u]& \bullet }, with a $DC4$-violating tuple
$(\delta(q_0, c)$, $e_5$, $\varepsilon$, $e_1$, $E_2)$, that in turn
requires attachment of $e_5$ to $E_2$, in order to enforce $DC4$.

If in this example, the order of $e_2$ and $b$ was reverse, i.e.,
the task automaton was $A_S^{\prime}$: \xymatrix@R=0.1cm{ \bullet
&\bullet \ar[l]_{d} &\bullet \ar[l]_{e_1} &\bullet
\ar[l]_{e_5}\ar[dl]^{e_1} &\bullet
\ar[l]_{c}\ar[r]^{e_1}\ar[dr]_{a}& \bullet \ar[r]^{a}& \bullet
\ar[r]^{e_2} & \bullet
\ar[r]^{b} & \bullet\\
&&\bullet &&\ar[u]& \bullet \ar[r]_{e_3}& \bullet} with $E_1$ =
$\{a, b, c, d, e_1, e_3$, $e_5\}$, $E_2$
             = $\{a, b, c, d, e_2\}$.
Then as it was shown in Example \ref{enforcing DC4 not by event
exclusion}, the $DC4$-violating tuple $(q_0, e_1, \varepsilon$, $a,
E_2)$ could not be dealt with exclusion of $a$ from $E_2$, due to
$EF2$, neither by inclusion of $e_1$ into $E_2$ (since as mentioned
above, it generates a new violation of $DC4$ that consequently
requires another inclusion of $e_5$ into $E_2$ to satisfy $DC4$).
\end{example}

\begin{remark}
Both Lemmas \ref{Resolution of DC3 violation} and \ref{Resolution of
DC4 violation} provide sufficient conditions for resolving the
violations of $DC3$ and $DC4$, respectively. They do not however
provide the necessary solutions, neither the optimal solutions, as
illustrated in the following example. We will show that for $DC3$
and $DC4$, in general one requires to search exhaustively to find
the optimal sequence of enforcing tuples, to have minimum number of
link additions. In this sense, instead of exhaustive search for
optimal solution, it is reasonable to introduce sufficient
conditions to provide a trackable procedure for a feasible solution
to make an automaton decomposable.
\end{remark}
\begin{example} Consider a task automaton
$A_S$:\\ \xymatrix@R=0.1cm{ \bullet &\bullet \ar[l]_{e_6} &\bullet
\ar[l]_{b} &\bullet \ar[l]_{e_5} &\bullet
\ar[l]_{e_7}&&\ar[d]&\bullet
 \ar[r]^{e_5}& \bullet \ar[r]^{e_3}&
\bullet \ar[r]^{a}& \bullet \ar[r]^{e_2}& \bullet
\\
&&&&& \bullet \ar[ul]_{e_1} \ar[dl]^{e_7}&\bullet
\ar[l]_{c}\ar[ur]^{e_1} \ar[dr]_{e_3}\\
\bullet & \bullet \ar[l]_{e_8} &\bullet \ar[l]_{b} &\bullet
\ar[l]_{e_1} &\bullet \ar[l]_{e_5}& & &  \bullet \ar[r]^{e_5}&
\bullet \ar[r]^{e_1}& \bullet \ar[r]^{a}& \bullet \ar[r]^{e_4}&
\bullet}\\
 with local event sets $E_1 = \{a, b, c, e_1, e_3, e_5, e_7\}$ and
 $E_2 = \{a, b, c, e_2, e_4, e_6, e_8\}$. $A_S$ is undecomposable
 due to $DC3$-violating tuples $(e_1e_5e_3ae_2, e_3e_5e_1ae_4, a, E_1$,
 $E_2)$ and $(e_1e_7e_5be_6$, $e_7e_5e_1be_8$, $a$, $E_1$, $E_2)$ and
 $DC4$-violating tuples $(q_0$, $e_1e_5e_3$, $e_3e_5e_1$, $a$, $E_2)$ and
 $(\delta(q_0, c), e_1e_7e_5, e_7e_5e_1, b, E_2)$. According to
 Lemmas \ref{Resolution of DC3 violation} and \ref{Resolution of
DC4 violation}, two enforcing tuples $\{\{e_1, e_3\}, E_2\}$ and
$\{\{e_1, e_7\}, E_2\}$ remove all violations of $DC3$ and $DC4$.
However, this solution is not unique, nor optimal, as the enforcing
tuple $\{\{e_1, e_5\}, E_2\}$ enforced $DC3$ and $DC4$ with minimum
number of added communication links.
\end{example}

\subsection{Exhaustive search for optimal decompozabilization}
Another difficulty is that enforcing the decomposability conditions
using link deletion is limited to passivity and $EF1$-$EF4$, and
after deletions of redundant links (that are passive and their
deletion respect $EF1$-$EF4$), the only way to make the automaton
decomposable is to establish new communication links. Addition of
new links, on the other hand, may lead to new violations of $DC3$ or
$DC4$ (as illustrated in Examples \ref{Conflicting DC3 with EF1, EF2
and EF4} and \ref{Enforcing DC4 leading to new violations}), and in
turn may introduce new violations. It means that, in general,
resolution of decomposability conditions can dynamically result in
new violations of decomposability conditions, as it is elaborated in
the following example.

\begin{example}\label{Dynamic violation of DC3 and DC4}
Consider the task automaton $A_S$: \\\xymatrix@R=0.3cm{ &\bullet
&\bullet \ar[l]_{e_{10}} &\bullet \ar[l]_{d} &\bullet \ar[l]_{e_2}
&\bullet \ar[l]_{e_6}
 \ar[r]^{e_2}& \bullet &\bullet
\\
\bullet &\bullet \ar[l]_{e_{12}} &\bullet \ar[l]_{b} &\bullet
\ar[l]_{e_2} &\bullet \ar[l]_{e_4}\ar[dl]^{e_2}&\bullet \ar[l]^{a}
 \ar[r]^{f}\ar[dr]_{e_1}\ar[u]^{c}&\bullet
 \ar[r]_{e_8}\ar[ur]^{e_4}& \bullet \ar[r]_{e_4}&\bullet \ar[r]_{g}
 &\bullet\\
&&&\bullet &&\ar[u] & \bullet  \ar[r]_{e_2} & \bullet \ar[r]_{e_3} &
\bullet \ar[r]_{e_5}&\bullet }
 with local event sets $E_1 = \{a, b, c, d, f, g, e_1, e_3, e_5\}$ and
 $E_2 = \{a, b, c,  d, f, g, e_2, e_4, e_6, e_8, e_{10}, e_{12}\}$.
 This automaton is undecomposable due to $DC2$-violating event pairs
 $\{(e_1, e_2)$, $(e_2, e_3)\}$ with the corresponding enforcing
 tuples $\{e_1, E_2\}$, $\{e_3, E_2\}$ and $\{e_2, E_1\}$ and with
 the following possible sequences:
\begin{enumerate}
 \item $\{e_1, E_2\}$; $\{e_3, E_2\}$:  in this case $A_S$ becomes
decomposable, without emerging new violations of decomposability
conditions;
\item $\{e_1, E_2\}$; $\{e_2, E_1\}$; $\{\{e_4, e_6\}, E_1\}$; $\{e_8, E_1\}$:
if after including $e_1$ in $E_2$, $e_2$ is included in $E_1$, then
two $DC4$-violating tuples $(\delta(q_0, a), \varepsilon, e_4, e_2,
E_1)$ and $(\delta(q_0, c), \varepsilon, e_6, e_2, E_1)$ emerge that
in turn require $\{e_4, e_6\}$ to be attached to $E_1$. Inclusion of
$e_4$ in $E_1$, on the other hand, introduces another
$DC4$-violating tuple $(\delta(q_0, f), \varepsilon, e_8, e_4, E_1)$
that calls for attachment of $e_8$ to $E_1$; similarly
\item $\{e_3, E_2\}$; $\{e_1, E_2\}$;
\item $\{e_3, E_2\}$; $\{e_2, E_1\}$; $\{\{e_4, e_6\}, E_1\}$; $\{e_8,
E_1\}$, and
\item $\{e_2, E_1\}$; $\{\{e_4, e_6\}, E_1\}$; $\{e_8, E_1\}$.
 \end{enumerate}
In this example, the first and the third sequences, i.e., $\{\{e_1,
e_3\}, E_2\}$ gives the optimal choice with minimum number of added
communication links, although initially $\{e_2, E_1\}$ sought to
offer the optimal solution.
\end{example}

Therefore, in general an optimal solution to Problem \ref{Design
Problem} will be obtained through an exhaustive search, using Lemmas
\ref{minimum spanning nodes-n agents-lemma}, \ref{Resolution of DC3
violation} and \ref{Resolution of DC4 violation}, as state in the
following algorithm.

\begin{algorithm}\label{Optimal Decomposabilization Algorithm}
\ \begin{enumerate}
\item For any local event set, exclude any passive event whose exclusion respects $EF1$-$EF4$;
\item identify all $DC1\&2$-violating tuples, $DC3$-violating tuples
and $DC4$-violating tuples and their respective enforcing tuples;
\item among all enforcing tuples, find the one that corresponds to
the most violating tuples;
\item if applying of the enforcing tuples with maximum number of
violating tuples, does not impose new violations of $DC3$ or $DC4$,
then apply it, go to Step $3$ and continue until there is no
violating tuples; otherwise, do the exhaustive search to find the
sequence of link additions with minimum number of added links.
\item end.
\end{enumerate}
\end{algorithm}
\begin{lemma}\label{Optimal Decomposabilization}
Consider a deterministic task automaton $A_S$ with local event sets
$E_i$ such that $E=\overset{n}{\underset{i=1}{\cup}}E_i$. If $A_S$
is not decomposable with respect to parallel composition and natural
projections $P_i$, $i = 1,...,n$, Algorithm \ref{Optimal
Decomposabilization Algorithm} optimally makes  $A_S$ decomposable,
with minimum number of communication links.
\end{lemma}

\begin{proof}
See the proof in the Attachment.
\end{proof}

\begin{remark}(Special case: Automata with mutual exclusive branches) \label{Mutual Exclusive Branches}
When branches of $A_S$ share no events (i.e. $\forall q\in Q$, $s,
s^{\prime} \in E^*$,  $\delta(q, s)!$, $\delta(q, s^{\prime})!$,
$s\nless s^{\prime}$, $s^{\prime} \nless s$: $s\cap s^{\prime} =
\emptyset$), due to definition of $DC3$ and $DC4$ in Lemma
\ref{Decomposability Corollary for n agents} $DC3$ and $DC4$ are
trivially satisfied, and moreover, since branches from any state
share no event, then Algorithm \ref{Optimal Decomposabilization
Algorithm} is reduced to Algorithm \ref{minimum spanning nodes-n
agents}.
\end{remark}

\subsection{Feasible solution for task decomposabilization}

As Example \ref{Dynamic violation of DC3 and DC4} showed that, in
general, additions of communication links may successively introduce
new violations of decomposability conditions, for which new links
should be established. Therefore, in general an optimal solution to
Problem \ref{Design Problem} requires an exhaustive search, using
Lemmas \ref{minimum spanning nodes-n agents-lemma}, \ref{Resolution
of DC3 violation} and \ref{Resolution of DC4 violation}. Moreover,
checking of $DC3$ and $DC4$ is a nontrivial task, while it has to be
accomplished initially as well as upon each link addition. It would
be therefore very tractable if we can define a procedure to make
$DC3$ and $DC4$ satisfied, without their examination. Following
result takes an automaton whose $DC1$ and $DC2$ are made satisfied
using Algorithm \ref{minimum spanning nodes-n agents}, and proposes
a sufficient condition to fulfill $DC3$ and $DC4$.

\begin{lemma}\label{Fixing DC3 and DC4}
Consider a deterministic automaton $A_S$, satisfying $DC1$ and
$DC2$. $A_S$ satisfies $DC3$ and $DC4$ if following steps are
accomplished on $A_S$:
\begin{enumerate}
\item $\forall s_1, s_2 \in E^*$, $s_1\nless s_2$, $s_2\nless s_1$, $q, q_1, q_2\in Q$, $\delta(q, s_1) =
q_1 \neq \delta(q, s_2) = q_2$, $[\nexists e_1, e_2 \in E, e_1e_2
\leqslant s_1$, $e_2e_1 \leqslant s_2$, $\forall t \in E^*$,
$\delta(q, e_1e_2t)! \Leftrightarrow \delta(q, e_2e_1t)!]$, $\exists
e\in s_1\cap s_2$,, then $\forall i\in loc(e)$, $\forall e^{\prime}
\in \{e_1\leqslant t_1, e_2\leqslant t_2\}$, $e^{\prime}$ appears
before $e$, include $e^{\prime}$ in $E_i$.
\item go to Step $1$ and continue until
$\forall s_1, s_2 \in E^*$, $s_1\nless s_2$, $s_2\nless s_1$, $q,
q_1, q_2\in Q$, $\delta(q, s_1) = q_1 \neq \delta(q, s_2) = q_2$,
$\exists e\in s_1\cap s_2$, $[\nexists e_1, e_2 \in E, e_1e_2
\leqslant s_1$, $e_2e_1 \leqslant s_2$, $\forall t \in E^*$,
$\delta(q, e_1e_2t)! \Leftrightarrow \delta(q, e_2e_1t)!]$, then
$\forall i\in loc(e)$,  $E_i$ contains the first events of $s_1$ and
$s_2$, that appear before $e$.
\end{enumerate}
\end{lemma}

\begin{proof}
See the proof in the Attachment.
\end{proof}

\begin{remark}
The condition in Lemma \ref{Fixing DC3 and DC4} intuitively means
that for any two strings $s_1$, $s_2$ from any state $q$, sharing an
event $e$, all agents who know this event $e$ will be able to
distinguish two strings, if they know the first event of each
string. The ability of those agents that know this event $e$ to
distinguish strings $s_1$ and $s_2$, prevents illegal interleavings
(to enforce $DC3$) and local nondeterminism (to satisfy $DC4$). The
significance of this condition is that it does not require to check
$DC3$ and $DC4$, instead provides a tractable (but more
conservative) procedure to enforce $DC3$ and $DC4$. The expression
$s_1\nless s_2$, $s_2 \nless s_1$ in the lemma, is to exclude the
pairs of strings that one of them is a substring of the other, as
their language product does not exceed from the strings of $A_S$,
provided $DC1$ and $DC2$. Moreover, the expression $[\nexists e_1,
e_2 \in E, e_1e_2 \leqslant s_1$, $e_2e_1 \leqslant s_2$, $\forall t
\in E^*$, $\delta(q, e_1e_2t)! \Leftrightarrow \delta(q, e_2e_1t)!]$
in this lemma excludes the pairs of strings $e_1e_2t$ and $e_2e_1t$
from any $q\in Q$ that have been already checked using $DC1$ and
$DC2$ and do not form illegal interleaving strings, and hance, do
not need to include $e_1$ in the local event sets of $e_2$ and vice
versa (see Example \ref{Design Example}).
\end{remark}

Combination of Lemmas \ref{minimum spanning nodes-n agents-lemma}
and \ref{Fixing DC3 and DC4} leads to the following algorithm as a
sufficient condition to make a deterministic task automaton
decomposable. Following algorithm uses Lemma \ref{minimum spanning
nodes-n  agents-lemma} to enforce $DC1$ and $DC2$ followed by Lemma
\ref{Fixing DC3 and DC4} to overcome the violations of $DC3$ and
$DC4$.

\begin{algorithm}\label{Feasible Decomposabilization Algorithm}
\ \begin{enumerate}
\item For a deterministic automaton $A_S$, with local event sets
$E_i$, $i = 1, \ldots, n$, $\forall E_i \in \{E_1, \ldots, E_n\}$,
$E_i^0 = E_i\backslash \{e\in E_i|e$ is passive in $E_i$ and
exclusion of $e$ from $E_i$ does not violate $EF1$-$EF4\}$;
\item  form the $DC1\&2$-Violating graph ; set $V^0(A_S) = V(A_S)$; $W^0(A_S) =
W(A_S)$; $G^0(A_S)=(W(A_S), \Sigma)$; k=1;
\item Among all events in the nodes in $W^{k-1}(A_S)$, find $e$ with
the maximum $|W_e^{k-1}(A_S, E_i^{k-1})|$, for all $E_i^{k-1}\in
\{E_1^{k-1}, \ldots, E_n^{k-1}\}$;
\item $E_i^k = E_i^{k-1} \cup \{e\}$;  and delete all edges from $e$ to $E_i^k$;
\item update $W_e^k(A_S, E_i)$ for all nodes of $G(A_S)$;
\item set $k=k+1$ and go to step $(3)$;
\item continue, until there exist no edges.
\item $\forall s_1, s_2 \in E^*$, $s_1\nless s_2$, $s_2\nless s_1$, $q, q_1, q_2\in Q$, $\delta(q, s_1) =
q_1 \neq \delta(q, s_2) = q_2$, $[\nexists e_1, e_2 \in E, e_1e_2
\leqslant s_1$, $e_2e_1 \leqslant s_2$, $\forall t \in E^*$,
$\delta(q, e_1e_2t)! \Leftrightarrow \delta(q, e_2e_1t)!]$, $\exists
e\in s_1\cap s_2$, then $\forall i\in loc(e)$, $\forall e^{\prime}
\in \{e_1\leqslant t_1, e_2\leqslant t_2\}$, $e^{\prime}$ appears
before $e$, include $e^{\prime}$ in $E_i$.
\item go to Step $1$ and continue until
$\forall s_1, s_2 \in E^*$, $s_1\nless s_2$, $s_2\nless s_1$, $q,
q_1, q_2\in Q$, $\delta(q, s_1) = q_1 \neq \delta(q, s_2) = q_2$,
$\exists e\in s_1\cap s_2$, $[\nexists e_1, e_2 \in E, e_1e_2
\leqslant s_1$, $e_2e_1 \leqslant s_2$, $\forall t \in E^*$,
$\delta(q, e_1e_2t)! \Leftrightarrow \delta(q, e_2e_1t)!]$, then
$\forall i\in loc(e)$, $E_i$ contains the first events of $s_1$ and
$s_2$, that appear before $e$.
\end{enumerate}
\end{algorithm}

Based on this formulation, a solution to Problem \ref{Design
Problem} is given as the following theorem.
\begin{theorem}\label{Feasible decomposition}
Consider a deterministic task automaton $A_S$ with local event sets
$E_i$ such that $E=\overset{n}{\underset{i=1}{\cup}}E_i$. If $A_S$
is not decomposable with respect to parallel composition and natural
projections $P_i$, $i = 1,...,n$, Algorithm \ref{Feasible
Decomposabilization Algorithm} makes $A_S$ decomposable. Moreover,
if after Step $7$, $DC3$ and $DC4$ are satisfied, then the algorithm
makes $A_S$ decomposable, with minimum number of communication
links.
\end{theorem}

\begin{proof}
After excluding the redundant shared events in the first step, the
algorithm enforces $DC1$ and $DC2$ in Steps $2$ to $7$, according to
Lemma \ref{minimum spanning nodes-n agents-lemma} and deals with
$DC3$ and $DC4$ in Steps $8$ and $9$, based on Lemma \ref{Fixing DC3
and DC4}.
\end{proof}
\begin{remark}
If after Step $7$, no violation of $DC3$ or $DC4$ is reported in the
automaton, then $A_S$ is made decomposable with minimum number of
added communication links; otherwise, the optimal solution can be
obtained through exhaustive search by examining the number of added
links for any possible sequence of enforcing tuples, using Lemmas
\ref{Resolution of DC3 violation} and \ref{Resolution of DC4
violation}, as it was presented in Lemma \ref{Optimal
Decomposabilization}. To avoid the exhaustive search the algorithm
provides a sufficient condition to enforce $DC3$ and $DC4$ in Steps
$8$ and $9$, according to Lemma \ref{Fixing DC3 and DC4}. The
algorithm terminates, due to finite number of states and events, and
the fact that at the worst case, when all events are shared among
all agents, the task automaton is trivially decomposable.
\end{remark}

\begin{example}\label{Design Example}
Consider a task automaton

$A_S$: \xymatrix@C=0.6cm{
 &\bullet &\bullet \ar[l]_{e_6}&&&&\bullet \ar[r]^{d} \ar[dr]^{e_{10}} &\bullet \ar[r]^{e_{10}} &\bullet\\
 & \bullet &\bullet \ar[l]_{e_7} &\bullet
 \ar[l]_{e_6}\ar[ul]_{e_7}&& & \bullet
 \ar[u]^{a}&\bullet \ar[r]^{d}&\bullet\\
 \bullet & \bullet\ar[l]_{e_9} & \bullet\ar[l]_{c} & \bullet \ar[l]_{e_2}\ar[dl]^{e_8}&
 \bullet \ar[l]_{b}
 \ar[ul]_{a}\ar[r]^{f}&\ar[ur]^{e_{11}}\ar[rr]^{e_2}\ar[dr]_{e_3}&&\bullet
 \ar[r]^{e_1} \ar[dr]_{e_5}&\bullet\\
\bullet & \bullet\ar[l]_{e_{12}} & \bullet\ar[l]_{e_2}& \ar[ur]&&&
\bullet \ar[ur]^{e_4}&&\bullet}\\
            with local event sets $E_1 = \{a, b, c, d, f, e_1, e_3,
e_5, e_7, e_9, e_{11}\}$ and $E_2 = \{a, b, c, d, f, e_2, e_4, e_6,
e_8,\\ e_{10}, e_{12}\}$, with the communication pattern $2 \in
snd_{a, b, c, d}(1)$ and no more communication links. This task
automaton is not decomposable, due to the set of $DC1\&2$-violating
tuples $\{e_1, e_2\}$, $\{e_1, e_4\}$, $\{e_2, e_3\}$, $\{e_2,
e_5\}$, $\{e_3, e_4\}$, $\{e_4, e_5\}$, $DC3$-violating tuples
$(e_{11}ade_{10}, ae_7e_6, a, E_1, E_2)$, $(e_{11}ade_{10}$,
$ae_6e_7, a, E_1, E_2)$, $(e_{11}ae_{10}d, ae_7e_6, a, E_1, E_2)$,
$(e_{11}ae_{10}d, ae_6e_7, a, E_1, E_2)$ and $DC4$-violating tuple
$(q_0, e_{11}, \varepsilon, a, E_2)$. There is also one event $d$
that is redundantly shared with $E_2$ as $d$ is passive in $E_2$ and
its exclusion respects $EF1$-$EF4$. Therefore, at the first step,
the algorithm excludes $d$ from $E_2$.

Next step is to construct the $DC1\&2$-Violating graph and remove
its edges by sharing one node from each edge. The set of
$DC1\&2$-Violating event pair is obtained as $V^0(A_S) = \{\{e_1,
e_2\}, \{e_1, e_4\},\\ \{e_2, e_3\}, \{e_2, e_5\}, \{e_3, e_4\},
\{e_4, e_5\}\}$ with $W^0(A_S) = \{e_1, e_2, e_3, e_4, e_5\}$.  It
can be seen that the private events $d, e_6, e_7, e_8, e_9, e_{10},
e_{11}, e_{12}$, and shared events $a, b, c, f$ are not included in
$W^0(A_S)$ as they have no contribution in violation of $DC1$ and
$DC2$. The $DC1\&2$-Violating graph is shown in Figure
\ref{GraphAlgorithmFig1}(a).
\begin{figure}[ihtp]
      \begin{center}
     \includegraphics[width=0.6\textwidth]{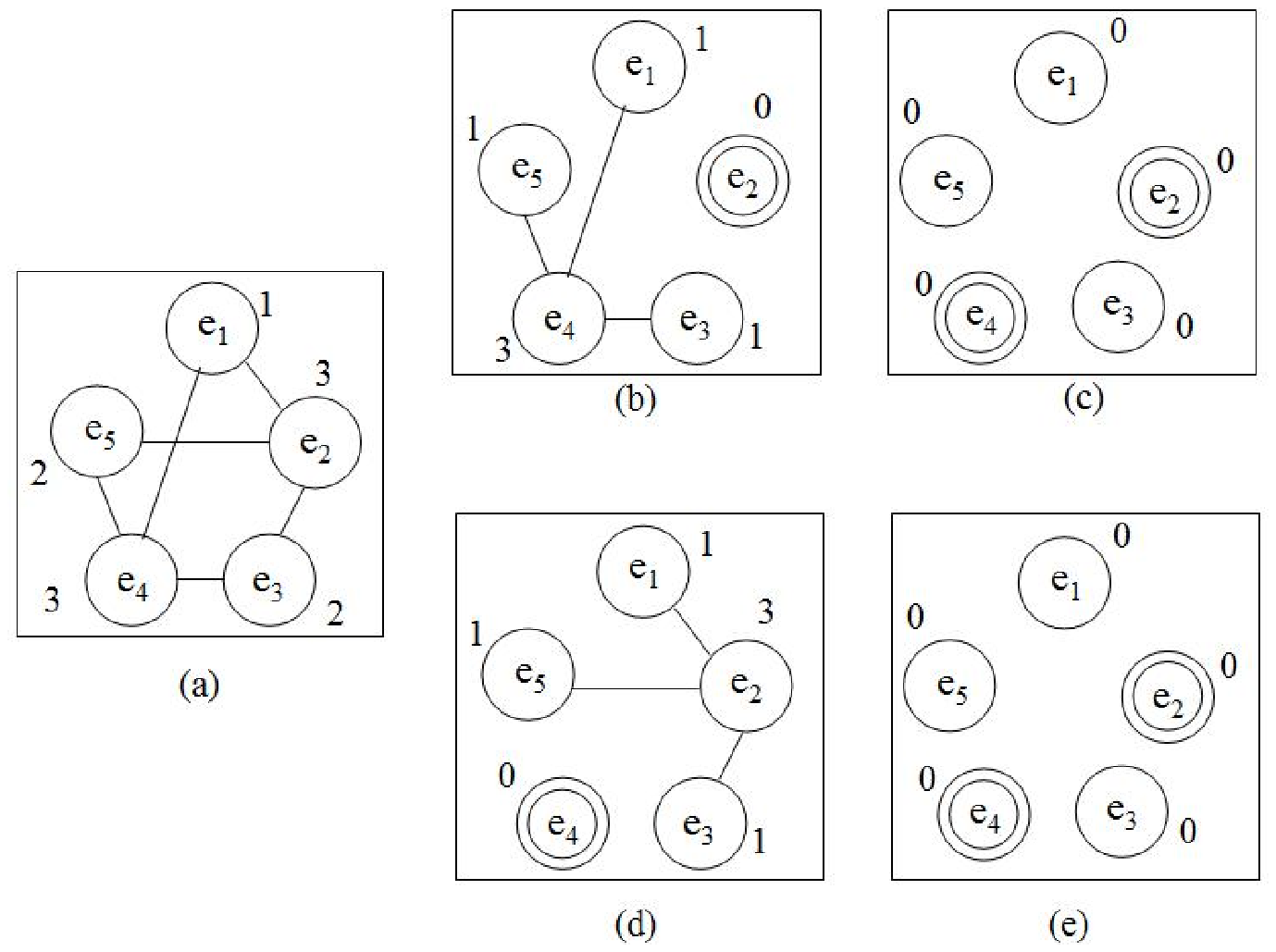}
        \caption{Illustration of enforcing $DC1$ and $DC2$ in Example \ref{Design
        Example},
        using Algorithm \ref{Feasible Decomposabilization Algorithm}.}
         \label{GraphAlgorithmFig1}
        \end{center}
      \end{figure}

The maximum $|W_e^{k-1}(A_S, E_i^{k-1})|$ is formed by $\{e_2,
e_4\}$ with respect to $E_1$ (here, since the system has only two
local event sets $|W_e^{k-1}(A_S, E_i^{k-1})|$ coincides to the
valency of $e$ in the graph). Marking $e_2$, including it to $E_1$
($E_1^1 = \{a, b, c, d, e_1, e_3, e_5, e_7, e_9, e_{11}, e_2\}$) and
removing its incident edges to $E_1$ and updating the $|W_e^k(A_S,
E_i^k)|$ (valencies) are shown in Figure
\ref{GraphAlgorithmFig1}(b). The next step will include $e_4$ in
$E_1$ ($E_1^2 = \{a, b, c, d, e_1, e_3, e_5, e_7, e_9, e_{11}, e_2,
e_4\}$) with the highest $|W_e^k(A_S, E_i^k)|$ and removing its
incident edges to $E_1$ and updating the $|W_e^k(A_S, E_i^k)|$ will
accomplish enforcing of $DC1$ and $DC2$ upon Step $7$, as it is
illustrated in Figure \ref{GraphAlgorithmFig1} (c). If from the
first stage $e_4$ was chosen instead of $e_2$, the procedure was
similarly performed as depicted in Figures \ref{GraphAlgorithmFig1}
(d) and (e), resulting the same set of private events $\{e_2, e_4\}$
to be shared with $E_1$. Inclusion of $e_2$ in $E_1$, however,
introduces a new $DC4$-violating tuple $(\delta(q_0, b),
\varepsilon, e_8, e_2, E_1)$ that will be automatically overcome in
Step $8$ by sharing $e_8\in s_1 = e_8e_2e_{12}$ (as $s_1=
e_8e_2e_{12}$ together with $s_2 = e_2ce_9$ evolve from $\delta(q_0,
b)$, sharing $e_2 \in s_1\cap s_2$) in all local event sets of
$e_2$, i.e., by including $e_8$ into $E_1$. Similarly, inclusion of
$e_{11}$ in $E_2$ overcomes $DC4$-violating tuple $(q_0, e_{11},
\varepsilon, a, E_2)$. It is worth noting that the expression
``$\nexists e_1, e_2 \in E, e_1e_2 \leqslant s_1$, $e_2e_1 \leqslant
s_2$, $\forall t\in E^*$, $\delta(q, e_1e_2t)!\Leftrightarrow
\delta(q, e_2e_1t)!$'' in Step $8$ prevents unnecessary inclusion of
$e_{10}$ in $E_1$ as well as $e_7$ in $E_2$ and $e_6$ in $E_1$
($e_6$ and $e_7$ satisfy $DC1$-$DC2$ and $e_{10}$ and $d$ satisfy
$EF1$-$EF2$). The algorithm terminates in this stages, leading to
decomposability of $A_S$, with $E_1^3 = \{a, b, c, d, e_1, e_3, e_5,
e_7, e_9, e_{11}, e_2, e_4, e_8\}$, $E_2^3 = E_2$, $E_2^3 = \{a, b,
c, e_2, e_4, e_6, e_8, e_{10}, e_{11}, e_{12}\}$.

\end{example}

\section{CONCLUSIONS}\label{CONCLUSION}
The paper proposed a method for task automaton decomposabilization,
applicable in top-down cooperative control of distributed discrete
event systems. This result is a continuation of our previous works
on task automaton decomposition
\cite{Automatica2010-2-agents-decomposability,
TAC2011-n-agents-decomposability}, and fault-tolerant cooperative
tasking \cite{Automatica2011-Fault-tolerant}, and investigates the
follow-up question to understand that how an originally
undecomposable task automaton can be made decomposable, by modifying
the event distribution among the agents.

First, using the decomposability conditions the sources of
undecomposability are identified and then a procedure was proposed
to establish new communication links in order to enforce the
decomposability conditions. 
To avoid the exhaustive search and the difficulty of checking of
decomposability conditions in each step, a feasible solution was
proposed as a sufficient condition that can make any deterministic
task automaton decomposable.
%

\section{APPENDIX}
\subsection{Proof of Lemma $\ref{minimum spanning nodes-n agents-lemma}$}
Following lemma will be used during the proof.
\begin{lemma}\label{Comparison of decreasing chains}
Consider two non-increasing chains $a_i$, $b_i$, $i=1,...,N$, such
that $a_1\geq a_2\geq ... \geq a_N>0$, $b_1\geq b_2\geq ... \geq
b_N>0$. Then $\overset{N}{\underset{i=1}{\Sigma}}a_i<
\overset{N}{\underset{i=1}{\Sigma}}b_i$ implies that $\exists k\in
\{ 1,..., N\}$ such that $a_k<b_k$.
\end{lemma}
\begin{proof}
Suppose by contradiction that
$\overset{N}{\underset{i=1}{\Sigma}}a_i<
\overset{N}{\underset{i=1}{\Sigma}}b_i$, but, $\nexists k\in \{
1,..., N\}$ such that $a_k < b_k$. Then, $\forall k \in \{1,...,N\}:
a_k\geq b_k$. Therefore, since $a_k, b_k >0, \forall k \in
\{1,...,N\}$, it results in $\overset{N}{\underset{i=1}{\Sigma}}a_i
\geq \overset{N}{\underset{i=1}{\Sigma}}b_i$ which contradicts to
the hypothesis, and the proof is followed.
\end{proof}
Now, we prove Lemma \ref{minimum spanning nodes-n agents-lemma} as
follows. In each iteration $k$ for the event $e$ and local event set
$E_i$ with maximum $|W_e^{k-1}(A_S, E_i^{k-1})|$, all edges from $e$
to $E_i$ are deleted. Denoting the set of deleted edges in $k-th$
iterations by $\Delta \Sigma^k$, in each iteration $k$, some
elements of $\Sigma^{k-1}$ are moved into $\Delta \Sigma^k$ until
after $K$ iterations, there is no more elements in $\Sigma^K$ to be
moved into a new set. This iterative procedure leads to a
partitioning of $\Sigma$ by $\{\Delta \Sigma ^k\}_{k = 1}^K$, as
$\{\Delta \Sigma^k\}\cap \{\Delta \Sigma^l\}=\emptyset$, $\forall
k,l=\{1,...,K\}, k\neq l$ and $\overset{K}{\underset{k=1}{\cup }}\
\Delta \Sigma ^k = \Sigma$. The latter equality leads to
\begin{equation}\label{partitioning-Size-K}
\overset{K}{\underset{k=1}{\Sigma }}|\Delta \Sigma ^k|=|\Sigma|
\end{equation}
Now, we want to prove that
\begin{equation}\label{partitioning-Max-Size}
|\Delta \Sigma ^k|=|\Delta \Sigma^k|_{max}, \forall k\in\{1,...,K\}
\Rightarrow K=K_{min}
\end{equation}

Here, $K$ is the total number of iterations that is also equal to
the number of added communication links to remove violations of
$DC1$ and $DC2$. In this sense, $K$ is desired to be minimized.

The proof of (\ref{partitioning-Max-Size}) is by contradiction as
follows. Suppose that $|\Delta \Sigma ^k|=|\Delta \Sigma ^k|_{max}$,
$\forall k\in\{1,...,K\}$, but, $K\neq K_{min}$, i.e., there exists
another partitioning $\{\Delta^{\prime} \Sigma ^k\}_{k =
1}^{K^{\prime}}$, with $K^{\prime} < K$ partitions, leading to
\begin{equation}\label{partitioning-Size-K-prime}
\overset{K^{\prime}}{\underset{k=1}{\Sigma }}|\Delta ^{\prime}
\Sigma ^k|=|\Sigma|
\end{equation}
In this case, from (\ref{partitioning-Size-K}) and
(\ref{partitioning-Size-K-prime}), we have
\begin{equation}\label{expand the equality of partitions}
\overset{K}{\underset{k=1}{\Sigma }}|\Delta \Sigma^k| =
\overset{K^{\prime}}{\underset{k=1}{\Sigma }}|\Delta \Sigma ^k| +
\overset{K}{\underset{k = K^{\prime}+1}{\Sigma }}|\Delta \Sigma ^k|=
\overset{K^{\prime}}{\underset{k = 1}{\Sigma }}|\Delta ^{\prime}
\Sigma ^k|.
\end{equation}
Since $|\Delta \Sigma ^k|>0$, $\forall k\in\{1,...,K\}$, then
$\overset{K}{\underset{k = K^{\prime}+1}{\Sigma }}|\Delta \Sigma ^k|
> 0$, then, (\ref{expand the equality of partitions}) results in
\begin{equation}\label{comparison of partitions}
\overset{K^{\prime}}{\underset{k = 1}{\Sigma }}|\Delta \Sigma ^k| <
\overset{K^{\prime}}{\underset{k = 1}{\Sigma }}|\Delta ^{\prime}
\Sigma ^k|.
\end{equation}
Moreover, since $|\Delta \Sigma ^k|>0$, $|\Delta^{\prime} \Sigma
^k|>0$, $\forall k\in\{1,...,K\}$, then (\ref{comparison of
partitions}) together with Lemma \ref{Comparison of decreasing
chains} imply that $\exists k\in \{1,...,K^{\prime}\}\subseteq
\{1,...,K\}$, i.e., $|\Delta \Sigma ^k| < |\Delta ^{\prime} \Sigma
^k|$, i.e., $\exists k\in \{1,...,K\}$ such that $|\Delta \Sigma
^k|\neq |\Delta \Sigma
 ^k|_{max}$,
which contradicts to the hypothesis, and hence,
(\ref{partitioning-Max-Size}) is proven. Moreover, if automaton
$A_S$ has no violations of $DC3$ and $DC4$ before and during the
iterations, then the algorithm make it decomposable with the minimum
number of added communication links, since the problem of making
decomposable is reduced to optimal enforcing of $DC1$ and $DC2$.

\subsection{Proof for Lemma \ref{Resolution of DC3 violation}}
For any $DC3$-violating tuple $(s_1, s_2, a,  E_i, E_j)$, exclusion
of $a$ from $E_i$ or $E_j$, excludes $a$ from $E_i\cap E_j$, leading
to $p_{E_i\cap E_j}(s_1)$ and $p_{E_i\cap E_j}(s_1)$ do not start
with $a$, and hence $(s_1, s_2, a,  E_i, E_j)$ will no longer act as
a $DC3$-violating tuple.

For the second method in this lemma, firstly $\forall q\in Q$, $s_1,
s_2 \in E^*$, $\delta(q, s_1)!$, $\delta(q, s_2)!$, $p_{E_i\cap
E_j}(s_1)$ and $p_{E_i\cap E_j}(s_2)$ start with $a$, such that
$(s_1, s_2, a, E_i, E_j)$ is a $DC3$-violating tuple, $\exists b\in
(E_i\cup E_j)\backslash (E_i\cap E_j)$ such that $b$ appears before
$a$ in $s_1$ or $s_2$ (since $A_S$ is deterministic and $p_{E_i\cap
E_j}(s_1)$ and $p_{E_i\cap E_j}(s_2)$ start with $a$).

 Two cases are possible, here: $b$
appears in only one of the strings $s_1$ or $s_2$; or $b$ appears in
both strings. If $b$ appears before $a$ in only of the strings, then
without loss of generality, assume that $b$ belongs to only $s_1$,
and hence, $\exists q, q_1, q_2, q^{\prime}_1, q^{\prime\prime}_1\in
Q_i\times Q_j$, $\omega_1, \omega_2 \in [(E_i\cup E_j)\backslash
(E_i\cap E_j)]^*$, $\omega ^{\prime}_1 \in (E_i\cup E_j)^*$, $a \in
E_i\cap E_j$ such that $\delta_{i,j}(q, \omega_1) = q^{\prime}_1$,
$\delta_{i,j}(q^{\prime}_1, b) = q^{\prime\prime}_1$,
$\delta_{i,j}(q^{\prime\prime}_1, \omega^{\prime}_1) = q_1$,
$\delta_{i,j}(q_1, a)!$, $\delta_{i,j}(q, \omega_2) = q_2$,
$\delta_{i,j}(q_2, a)!$, where, $\delta_{i,j}$ is the transition
relation in $P_i(A_S)||P_j(A_S)$. Now, due to synchronization
constraint in parallel composition, inclusion of $b$ in $E_i\cap
E_j$ means that $([q^{\prime\prime}_1], y)$ and $(x,
[q^{\prime\prime}_1]_j)$ are accessible in $P_i(A_S)||P_j(A_S)$ only
if $y = [q^{\prime\prime}_1]_j$ and $x = [q^{\prime\prime}_1]_i$,
respectively. This means that $([q_1]_i, [q_2]_j)$ and $([q_2]_i,
[q_1]_j)$ are not accessible in $P_i(A_S)||P_j(A_S)$, and hence,
$p_i(s_1)|p_j(s_2)$ and $p_i(s_2)|p_j(s_1)$ cannot evolve after $a$,
and therefore, do not generate illegal strings out of the original
strings, implying that $(s_1, s_2, a,  E_i, E_j)$ will no longer
remain a $DC3$-violating tuple.

On the other hand, if $b$ appears before $a$, in both strings $s_1$
and $s_2$, then $\exists q, q_1, q_2, q^{\prime}_1,
q^{\prime\prime}_1, q^{\prime}_2, q^{\prime\prime}_2\in Q_i\times
Q_j$, $\omega_1, \omega_2 \in [(E_i\cup E_j)\backslash (E_i\cap
E_j)]^*$, $\omega ^{\prime}_1, \omega ^{\prime}_2 \in (E_i\cup
E_j)^*$, $a \in E_i\cap E_j$ such that $\delta_{i,j}(q, \omega_1) =
q^{\prime}_1$, $\delta_{i,j}(q^{\prime}_1, b) = q^{\prime\prime}_1$,
$\delta_{i,j}(q^{\prime\prime}_1, \omega^{\prime}_1) = q_1$,
$\delta_{i,j}(q_1, a)!$, $\delta_{i,j}(q, \omega_2) = q^{\prime}_2$,
$\delta_{i,j}(q^{\prime}_2, b) = q^{\prime\prime}_2$,
$\delta_{i,j}(q^{\prime\prime}_2, \omega^{\prime}_2) = q_2$,
$\delta_{i,j}(q_2, a)!$, that leads to accessibility of
$([q_1^{\prime}]_i, [q_2^{\prime}]_j)$ and $([q_2^{\prime}]_i,
[q_1^{\prime}]_j)$ as well as $([q_1]_i, [q_2]_j)$ and $([q_2]_i,
[q_1]_j)$ in $P_i(A_S)||P_j(A_S)$, that means that although $(s_1,
s_2, a,  E_i, E_j)$ is no longer a $DC3$-violating tuple, $(s_1,
s_2, b,  E_i, E_j)$ emerges as a new $DC3$-violating tuple.

In this case (when $\nexists b\in (E_i\cup E_j)\backslash (E_i\cap
E_j)$ that appears before $a$ in only one of the strings $s_1$ or
$s_2$), instead of inclusion of $b$ in $E_i\cap E_j$, if two
different events that appear before $a$ in strings $p_{E_i\cup
E_j}(s_1)$ and $p_{E_i\cup E_j}(s_1)$ are attached to $E_i\cap E_j$,
it leads to $\exists q, q_1, q_2, q_3, q_4\in Q_i\times Q_j$,
$\omega_1, \omega_2, \omega^{\prime}_1, \omega^{\prime}_2 \in
[(E_i\cup E_j)\backslash (E_i\cap E_j)]^*$, $e_1, e_2, a \in E_i\cap
E_j$ such that $\delta_{i,j}(q, \omega_1e_1) = q_1$,
$\delta_{i,j}(q_1, \omega^{\prime}_1) = q_3$, $\delta_{i,j}(q_3,
a)!$, $\delta_{i,j}(q, \omega_2e_2) = q_2$, $\delta_{i,j}(q_2,
\omega^{\prime}_2) = q_4$, $\delta_{i,j}(q_4, a)!$. Consequently,
due to synchronization constraint in parallel composition,
$([q_1]_i, [q]_j)$, $([q]_i, [q_1]_j)$, $([q_2]_i, [q]_j)$ and
$([q]_i, [q_2]_j)$, and hence, $([q_3]_i, [q_4]_j)$ and $([q_4]_i,
[q_3]_j)$ are not accessible in $P_i(A_S)||P_j(A_S)$, i.e., no more
$DC3$-violating tuples form on strings $s_1$ and $s_2$.

\subsection{Proof for Lemma \ref{Resolution of DC4 violation}}
For any $DC4$-violating tuple $(q, t_1, t_2, e, E_i)$, with $q, q_1,
q_2 \in Q$, $t_1, t_2 \in (E\backslash E_i)^*$, $e\in E_i$,
$\delta(q, t_1) = q_1 \neq \delta(q, t_2) = q_2$, exclusion of $e$
from $E_i$ leads to $p_i(e) = \varepsilon$, and $p_i(t_1e) = p_i
(t_2e) = \varepsilon$, $[q]_i = [\delta(q_1, e)]_i = [\delta(q_2,
e)]_i$, and hence, $(q, t_1, t_2, e, E_i)$ will no longer behave as
a $DC4$-violating tuple. However, it should be noted that it may
cause another nondeterminism on an event after $e$, and this event
exclusion is allowed only if $e$ is passive in $E_i$ and the
exclusion does not violate $EF1-EF4$.

For the second method, i.e., event inclusion, if $\exists e^{\prime}
\in (t_1\cup t_2)\backslash (t_1\cap t_2)$, then without loss of
generality, assume that $e^{\prime} \in t_1\backslash t_2$ such that
$\exists q, q_1, q_2, q^{\prime}_1, q^{\prime\prime}_1\in Q$, $t_1,
t_2 \in (E\backslash E_i)^*$, $e\in E_i$, $\delta(q, t_1) = q_1 \neq
\delta(q, t_2) = q_2$, $\delta(q^{\prime}_1, e^{\prime}) =
q^{\prime\prime}_1$. In this case, inclusion of $e^{\prime}$ in
$E_i$ leads to $p_i(t_1e) = e^{\prime}e$, while $p_i(t_2e) = e$, and
therefore, $[q_1]_i = [q^{\prime\prime}_1]_i \neq [q_2]_i$, i.e., in
$P_i(A_S)$, $t_1$ and $t_2$ will no longer cause a nondeterminism on
$e$ from $q$, and accordingly, $(q, t_1, t_2, e, E_i)$ will not
remain a $DC4$-violating tuple.

 If however $\nexists e^{\prime} \in
(t_1\cup t_2)\backslash (t_1\cap t_2)$, i.e., $\forall e^{\prime}
\in (t_1\cup t_2)$, $e^{\prime}\in(t_1\cap t_2)$, then inclusion of
any such $e^{\prime}$ generates a $DC4$-violating tuple $(q, t_1,
t_2, e^{\prime}, E_i)$. In this case, Lemma \ref{Resolution of DC4
violation} suggests to take two different events that appear before
$e$, one from $t_1$ and the other from $t_2$, and include them into
$E_i$ such that $\exists q, q_1, q_2, q^{\prime}_1, q^{\prime}_2 ,
q^{\prime\prime}_1, q^{\prime\prime}_2\in Q$, $e_1\in t_1, e_2\in
t_2$, $e_1 \neq e_2$, $\delta(q, t_1) = q_1\neq \delta(q, t_2) =
q_2$, $\delta(q^{\prime}_1, e_1) = q^{\prime\prime}_1$,
$\delta(q^{\prime}_2, e_2) = q^{\prime\prime}_2$. Thus, including
$e_1$ and $e_2$ in $E_i$ results in $p_i(t_1) = e_1$, $p_i(t_2) =
e_2$, $\delta_i ([q]_i, t_1)! = [q_1]_i \neq \delta_i([q]_i, t_2)=
[q_2]_i$, meaning that $(q, t_1, t_2, e, E_i)$ is not a
$DC4$-violating tuple anymore.

\subsection{Proof for Lemma \ref{Optimal Decomposabilization}}
 The algorithm starts with excluding events from local event
sets in which the events are passive and their exclusion do not
violate $EF1$-$EF4$. From this stage onwards the decomposability
conditions are no longer allowed to be enforced by link deletion,
whereas the algorithm removes the violations of decomposability
conditions by establishing new communication links. Next, the
algorithm applies violating tuples in the order of corresponding
number of violating tuples. If no new violations of decomposability
conditions emerge during conducting of enforcing tuples, then the
algorithm decomposes the task automaton with minimum number of
 communication links, similar to the proof of Lemma $\ref{minimum spanning nodes-n
 agents-lemma}$, since iterations partition the set of violating
 tuples, and applying of enforcing tuples (based on Lemmas \ref{minimum spanning nodes-n
 agents-lemma}, \ref{Resolution of DC3 violation} and \ref{Resolution of DC4
violation}) with maximum number of violating tuples in each
iteration gives maximum number of resolutions per link addition that
leads to the minimum number of added communication links. The
algorithm will terminate due to finite number of states and events
and at the worst case all events are shared among all agents to make
the automaton decomposable.

\subsection{Proof for Lemma \ref{Fixing DC3 and DC4}}
Denoting the expression , ``$\forall E_i, E_j \in \{E_1, \ldots,
E_n\}$, $i \neq j$, $a \in E_i \cap E_j$,  $s = t_1at_1^{\prime},
s^{\prime} = t_2at_2^{\prime}$, $p_{E_i\cap E_j}(t_1) = p_{E_i\cap
E_j}(t_2) = \varepsilon$'' as $A$, and the expression ``$\delta(q_0,
\mathop |\limits_{i = 1}^n p_i \left( {s_i } \right))!$ for any
$\forall \{s_1, \cdots, s_n\}\subseteq \tilde{L}(A_S)$, $s,
s^{\prime} \in \{s_1, \cdots, s_n\}$'' as $B$, the condition $DC3$
can be written as $A\Rightarrow B$. Now, if $\forall s_1, s_2 \in
E^*$, $s_1\nless s_2$, $s_2\nless s_1$, $q, q_1, q_2\in Q$,
$\delta(q, s_1) = q_1 \neq \delta(q, s_2) = q_2$, $[\nexists e_1,
e_2 \in E, e_1e_2 \leqslant s_1$, $e_2e_1 \leqslant s_2$, $\forall t
\in E^*$, $\delta(q, e_1e_2t)! \Leftrightarrow \delta(q,
e_2e_1t)!]$, $\exists e\in s_1\cap s_2$, any $e^{\prime} \in
\{e_1\leqslant t_1, e_2\leqslant t_2\}$, such that $e^{\prime}$
appears before $e$, is included in $E_i$, $\forall i\in loc(e)$, it
follows that $\forall E_i, E_j \in \{E_1, \ldots, E_n\}$, $i \neq
j$, $a \in E_i \cap E_j$,  $s = t_1at_1^{\prime}, s^{\prime} =
t_2at_2^{\prime}$, $\delta(q_0, s)! \neq \delta(q_0, s^{\prime})!$,
$a \in s\cap s^{\prime}$, then the first event of $t_1$ and $t_2$
belong to $E_i\cap E_j$, i.e., $A$ (the antecedent of $DC3$) becomes
false, and hence, $A\Rightarrow B$ ( $DC3$ ) holds true. Therefore,
the procedure in Lemma \ref{Fixing DC3 and DC4} gives a sufficient
conditions to make $DC3$ always true.

It is similarly a sufficient condition for $DC4$ as follows. Let the
expressions ``$\forall i\in\{1, \ldots, n\}$, $x, x_1, x_2 \in Q_i$,
$e\in E_i$, $t\in E_i^*$, $\delta(x, e) = x_1 \neq \delta(x, e) =
x_2$'' and ``$\forall t\in E_i^*: \delta(x_1, t)!\Leftrightarrow
\delta(x_2, t)!$'' to be denoted as $C$ and $D$, respectively. In
this case, $DC4$ can be expressed as $C\Rightarrow D$. Then, for a
deterministic automaton $A_S$, if $\forall s_1, s_2 \in E^*$,
$s_1\nless s_2$, $s_2\nless s_1$, $q, q_1, q_2\in Q$, $\delta(q,
s_1) = q_1 \neq \delta(q, s_2) = q_2$, $[\nexists e_1, e_2 \in E,
e_1e_2 \leqslant s_1$, $e_2e_1 \leqslant s_2$, $\forall t \in E^*$,
$\delta(q, e_1e_2t)! \Leftrightarrow \delta(q, e_2e_1t)!]$, $\exists
e\in s_1\cap s_2$, the first event of $s_1$ and $s_2$ are included
in all local event sets that contain $e$, it results in $\neg
C$(i.e., the antecedent of $DC4$ becomes false, and consequently,
$DC4$ becomes always true), since in such case $\forall E_i \in
\{E_1, \ldots, E_n\}$, $t_1, t_2 \in E^*$, $q, q_1, q_2 \in Q$,
$e\in E_i$, $\delta(q, t_1e) = q_1 \neq \delta(q, t_2e) = q_2$, then
$\neg[p_i(t_1) = p_i(t_2) = \varepsilon]$.

Expression ``$[\nexists e_1, e_2 \in E$, $e_1e_2 \leqslant s_1$,
$e_2e_1 \leqslant s_2$, $\forall t \in E^*$, $\delta(q, e_1e_2t)!
\Leftrightarrow \delta(q, e_2e_1t)!]$'', in Lemma \ref{Fixing DC3
and DC4} is to exclude those pairs of strings $s_1$ and $s_2$ that
start with $e_1e_2$ and $e_2e_1$, respectively, as they have been
already checked with $DC1$ and $DC2$ and their interleaving does not
impose illegal strings.
\bibliographystyle{IEEEtran}
\bibliography{Ref_13_arXiv}
\end{document}